\documentclass[english,11pt,a4paper]{article}
\usepackage[english]{babel}
\usepackage{latexsym,amssymb,amsmath,amsfonts,amscd,epsfig,amsthm,color,mathrsfs}

\usepackage{aliascnt}
\usepackage{graphicx}
\usepackage{tikz}
\usepackage{quantikz}
\usepackage{float}
\usepackage[left=2.3cm,right=2.3cm,top=2.3cm,bottom=2.3cm]{geometry}
\usepackage{ifthen}
\usepackage{authblk}
\usepackage{colonequals}
\usepackage{array}
\usepackage{tikz}
\usetikzlibrary{decorations.pathreplacing}

\makeatletter
\newtheorem*{rep@theorem}{\rep@title}
\newcommand{\newreptheorem}[2]{%
\newenvironment{rep#1}[1]{%
 \def\rep@title{#2 \ref*{##1}}%
 \begin{rep@theorem}}%
 {\end{rep@theorem}}}
\makeatother

\usepackage{hyperref}
\hypersetup{
    bookmarksnumbered=true, 
    unicode=false, 
    pdfstartview={FitH}, 
    pdftitle={QMA=QMA1 with an infinite counter}, 
    pdfauthor={Stacey Jeffery and Freek Witteveen}, 
    pdfsubject={}, 
    pdfcreator={}, 
    pdfproducer={}, 
    pdfkeywords={}, 
    pdfnewwindow=true, 
    colorlinks=true, 
    linkcolor=blue, 
    citecolor=blue, 
    filecolor=blue, 
    urlcolor=blue 
}

\usepackage[nameinlink]{cleveref}
\Crefname{section}{Section}{Sections}

\newtheorem{theorem}{Theorem}[section]

\newaliascnt{definition}{theorem}
\newtheorem{definition}[definition]{Definition}
\aliascntresetthe{definition}
\Crefname{definition}{Definition}{Definitions}

\newaliascnt{lemma}{theorem}
\newtheorem{lemma}[lemma]{Lemma}
\aliascntresetthe{lemma}
\Crefname{lemma}{Lemma}{Lemmas}

\newaliascnt{proposition}{theorem}
\newtheorem{proposition}[proposition]{Proposition}
\aliascntresetthe{proposition}
\Crefname{proposition}{Proposition}{Propositions}

\newaliascnt{corollary}{theorem}
\newtheorem{corollary}[corollary]{Corollary}
\aliascntresetthe{corollary}
\Crefname{corollary}{Corollary}{Corollaries}

\newaliascnt{claim}{theorem}

\aliascntresetthe{claim}
\Crefname{claim}{Claim}{Claims}

\newaliascnt{example}{theorem}

\aliascntresetthe{example}
\Crefname{example}{Example}{Examples}

\newaliascnt{conjecture}{theorem}

\aliascntresetthe{conjecture}
\Crefname{conjecture}{Conjecture}{Conjectures}

\newaliascnt{aside}{theorem}

\aliascntresetthe{aside}
\Crefname{aside}{Aside}{Asides}

\newaliascnt{remark}{theorem}
\newtheorem{remark}[remark]{Remark}
\aliascntresetthe{remark}
\Crefname{remark}{Remark}{Remarks}

\def\bigO#1{O\left(#1\right)}

\def\ket#1{{\lvert}#1\rangle}

\def\bra#1{{\langle}#1\rvert}
\def\braket#1#2{{{\langle}#1\vert}#2\rangle}
\def\abs#1{\left| #1 \right|}

\def\norm#1{\left\| #1 \right\|}
\DeclareMathOperator{\Tr}{Tr}

\DeclareMathOperator{\accept}{acc}

\newcommand{\CC}{\mathbb{C}}
\newcommand{\ZZ}{\mathbb{Z}}
\newcommand{\HH}{\mathcal{H}}

\newcommand{\incr}{[+1]}
\newcommand{\circincr}{+1}

\newcommand{\ketbra}[2]{|{#1}\rangle\!\langle{#2}|}
\newcommand{\ot}{\otimes}

\title{${\sf QMA}={\sf QMA}_1$ with an infinite counter}
\author[1,2]{Stacey Jeffery}
\author[1,3]{Freek Witteveen}
\affil[1]{QuSoft \& CWI, Amsterdam}
\affil[2]{University of Amsterdam}
\affil[3]{University of Copenhagen}
\date{}

\begin{document}

\maketitle

\vspace{-35pt}

\begin{abstract}
    A long-standing open problem in quantum complexity theory is whether ${\sf QMA}$, the quantum analogue of ${\sf NP}$, is equal to ${\sf QMA}_1$, its one-sided error variant. We show that ${\sf QMA}={\sf QMA}^{\infty}= {\sf QMA}_1^{\infty}$, where ${\sf QMA}_1^\infty$ is like ${\sf QMA}_1$, but the verifier has an infinite register, as part of their witness system, in which they can efficiently perform a shift (increment) operation. We call this register an ``infinite counter'', and compare it to a program counter in a Las Vegas algorithm. The result ${\sf QMA}={\sf QMA}^\infty$ means such an infinite register does not increase the power of ${\sf QMA}$, but does imply perfect completeness.

    By truncating our construction to finite dimensions, we get a ${\sf QMA}$-amplifier that only amplifies completeness, not soundness, but does so in significantly less time than previous ${\sf QMA}$ amplifiers. Our new construction achieves completeness $1-2^{-q}$ using $O(1)$ calls to each of the original verifier and its inverse, and $O(\log q)$ other gates, proving that ${\sf QMA}$ has completeness \emph{doubly} exponentially close to 1, i.e.~${\sf QMA}={\sf QMA}(1-2^{-2^r},2^{-r})$ for any polynomial $r$.
\end{abstract}

\section{Introduction}

The complexity class ${\sf QMA}$ is generally considered to be the quantum analogue of ${\sf NP}$, although there are several complexity classes that could claim this title, including the one-sided variant of ${\sf QMA}$, ${\sf QMA}_1$~(see \cite{gharibian23sevenFaces}).  ${\sf QMA}$ 
is the class of problems $(L_{\text{yes}},L_{\text{no}})\subseteq \{0,1\}^*\times\{0,1\}^*$ such that there exists a polynomial-time quantum verifier $V(x)$ such that (see \Cref{def:QMA} for a formal definition):
\begin{description}
    \item[Completeness:] For all $x\in L_{\text{yes}}$, there exists a witness $\ket{\psi}$ 
          such that $V(x)$ 
          accepts $\ket{\psi}$ with probability at least $c=2/3$.
    \item[Soundness:] For all $x\in L_{\text{no}}$, for all witnesses $\ket{\psi}$, $V(x)$ accepts with probability at most $s=1/3$.
\end{description}
Above, we can take any $s\geq 2^{-{\sf poly}(n)}$ and $c\leq 1-2^{-{\sf poly}(n)}$ such that $c-s\geq \frac{1}{{\sf poly}(n)}$, where $n$ is the length of the instance $x$, without changing the complexity class.
If we take $c=1$ (i.e. \emph{perfect completeness}), the class is ${\sf QMA}_1$, and it is not known whether this is the same class as ${\sf QMA}$.

${\sf QMA}$ is actually the quantum analogue of the class ${\sf MA}$, which is the same as ${\sf QMA}$, except that we restrict to classical circuits $V(x)$ and classical witnesses. The idea behind the name ${\sf MA}$, standing for ``Merlin-Arthur'', is that we can think of the verifier as representing a single-round protocol between a computationally unbounded \emph{prover}, Merlin, and a polynomial-time \emph{verifier} Arthur. Merlin tries to convince Arthur that $x\in L_{\text{yes}}$ by sending him a witness, and Arthur checks this witness using $V(x)$.

The class ${\sf QMA}$ is the more natural quantum analogue of ${\sf NP}$ than ${\sf QMA}_1$, since it makes sense to let quantum algorithms have two-sided error. Quantum algorithms tend to work imperfectly, and even a perfect quantum algorithm would likely exhibit some error once it is subject to noise in a real-world quantum computer. Moreover, an exact quantum algorithm in one universal quantum gate set will possibly become inexact when converted to another. So it is not practically interesting to restrict to quantum algorithms with no error, or even one-sided error. However, classically we have ${\sf MA}={\sf MA}_1$, raising the question of whether the same holds true for quantum, and if not, if there is some deep and interesting reason for this. It is even known to hold for ${\sf QMA}$ with classical witnesses -- i.e.~${\sf QCMA}={\sf QCMA}_1$~\cite{jordan2012QCMAisQCMA1}, as well as for a variant of ${\sf QMA}$ with a constant number of pre-shared EPR pairs \cite{kobayashi2013stronger} and for ${\sf preciseQMA}$, a modification of ${\sf QMA}$ with an exponentially small soundness-completeness gap \cite{fefferman2018complete}.

There are other reasons the class ${\sf QMA}_1$ is problematic. Because, as mentioned above, one-sided error is a gate-set-dependent property, there is no single definition of ${\sf QMA}_1$, but rather, it is always defined with respect to a particular choice of gate set (for a discussion of this, and progress towards a resolution, see~\cite{rudolph24universalQMA1}). One thing we do not want our complexity classes to be is sensitive to details as low-level as \emph{gate set}.
However, there are problems that are ${\sf QMA}_1$-hard~\cite{bravyi2006frustrationFree,gosset2013QMA13SAT,crichigno2024cliqueHomology,rudolph2025quantum2SAT,rudolph24universalQMA1}, so we also cannot simply declare this an irrelevant class and ignore it.
Showing that ${\sf QMA}={\sf QMA}_1$ would resolve the gate-set dependence of ${\sf QMA}_1$, as ${\sf QMA}$ does not suffer from gate set sensitivity.
There exists a quantum oracle relative to which ${\sf QMA}\neq {\sf QMA}_1$~\cite{aaronson09perfectCompQMA} (but this also separates ${\sf QCMA}$ from ${\sf QCMA}_1$ which are known to be equal).
This quantum oracle separation relies crucially on the dimension of the verifier's system being finite.

\subsection{Contributions}

\paragraph{Perfect-completeness of ${\sf QMA}^\infty$.} We make progress towards resolving the question of ${\sf QMA}$ vs.~${\sf QMA}_1$ by showing that they are equal when the prover and the verifier are allowed to use an infinite-dimensional register,\footnote{We use the convention $0\in\mathbb{N}$.} $\mathrm{span}\{\ket{d}:d\in\mathbb{N}\}$, on which the increment operation\footnote{We use the increment \emph{isometry} to keep the model minimal, but we could just as well have defined the model by letting the infinite-dimensional register have negative numbers as well, so that incrementing is \emph{unitary} (see \Cref{rmk:extensions}).}, $\ket{d}\mapsto \ket{d+1}$ can be efficiently performed, controlled on another qubit register, and we can efficiently check if $d=0$. Formally, we define ${\sf QMA}_1^\infty$ to be ${\sf QMA}_1$ with such verifiers (see \Cref{sec:infinite operations} for the formal definition). This means that the verifier can receive an infinite-dimensional witness, which is crucial for our results. We prove the following.
\begin{theorem}\label{thm:QMAisQMA1-inf}
    ${\sf QMA}={\sf QMA}_1^\infty$.
\end{theorem}
\noindent In \Cref{sec:main}, we show how to turn any ${\sf QMA}$ verifier into a verifier with perfect completeness that uses an infinite counter, showing that ${\sf QMA}\subseteq{\sf QMA}_1^\infty$ (see \Cref{cor:QMAinQMA-inf-exp-soundness}, where we also give a more precise statement, including mild requirements on the gate set with respect to which ${\sf QMA}_1^\infty$ is defined), and the full statement follows from \Cref{cor:QMA1-inf-is-QMA-inf-is-QMA}.

The extent to which \Cref{thm:QMAisQMA1-inf} resolves the open question depends entirely on the reasonableness of the extra power we have given the verifier in ${\sf QMA}_1^\infty$. {There are two sources of potential additional computational power: the verifier receives an infinite sized witness, and the verifier can potentially perform more powerful computations on the infinite register.}
An infinite-dimensional register on which the verifier can perform too powerful operations could potentially make it trivial to give any ${\sf QMA}$ verifier perfect completeness.
Indeed, one can also show perfect completeness if the verifier has an infinite-dimensional register that can encode a real angle $\theta\in [0,2\pi)$, as part of the witness, in such a way that the verifier could implement a $\theta$-rotation controlled on this register~\cite{aaronson09perfectCompQMA,nagaj2009FastAmpQMA}. The infinite-dimensional power we give our verifier is much weaker, first in that the dimension of its infinite-dimensional register is countably infinite, and second, in that we needn't ever read or control on the state of this infinite-dimensional register (aside from controlling on whether it is 0), but only shift its value by 1, which has a physically realistic interpretation.

In fact, every classical ${\sf BPP}$ (or even $\sf P$) machine does have such an infinite-dimensional register. One part of a Turing Machine's state is the position of the head on the infinite tape, which can be described by a natural number. At any step, this part of the state can increase (or decrease) by one, and it is also natural to imagine the head can efficiently know when it is in position 0 (for example, by trying and failing to move left). In this light, it does not seem unnatural to allow a ${\sf QMA}$ or ${\sf QMA}_1$ verifier the same power\footnote{In fact, an infinite counter could be quite a bit \emph{less} powerful than an infinite tape you can move up and down, reading and writing.
}
and denote the resulting complexity classes by ${\sf QMA}^\infty$ and ${\sf QMA}_1^\infty$ (see \Cref{sec:infinite operations} for formal definitions).
In fact, we show that such an infinite register does not give any extra power to ${\sf QMA}$ (see \Cref{sec:proof truncating witness} for a proof):
\begin{theorem}\label{thm:QMA-inf}
    ${\sf QMA}={\sf QMA}^\infty$.
\end{theorem}
\noindent So adding an infinite counter does not make $\sf QMA$ more powerful, but it does allow us to prove perfect completeness.
In other words, with an appropriate gate set, which happens to act on an infinite dimensional system but otherwise does not change the computational power of ${\sf QMA}$, the class ${\sf QMA}$ has perfect completeness, giving a partial resolution of the question whether ${\sf QMA}$ equals ${\sf QMA}_1$.

An infinite-dimensional Hilbert space of the form we allow is physically motivated, as it correctly describes certain physical systems, such as bosonic systems. The study of complexity classes analogous to $\sf BQP$ and $\sf QMA$ in bosonic systems has only recently been initiated~\cite{chabaud2025bosonic,upreti2025bounding}, and their relationship with standard classes remains unknown. We discuss this more in \Cref{sec:bosonic quantum computers}.

\paragraph{Doubly-exponential Completeness-amplification.} We can get back to the standard finite definition of ${\sf QMA}$ by truncating our construction, also truncating the witness to a finite dimensional subspace. This results in imperfect completeness. However, by truncating at large enough $d$, the completeness error $1-c$ can be made \emph{doubly-exponentially} small, with a verifier that is still polynomial (in the standard model, with no infinite-dimensional operations) improving the previously known result ${\sf QMA}={\sf QMA}(1-2^{-r},2^{-r})$ for any polynomial $r$, to
the following (proven in \Cref{sec:amplification}):
\begin{theorem}\label{thm:QMA-small-completeness-error}
    For any polynomial $r$, ${\sf QMA}={\sf QMA}(1-2^{-2^r},2^{-r})$.
\end{theorem}
\noindent A simple corollary is the ${\sf QMA}$-hardness of a larger class of problems, such as $k$-SAT with doubly-exponentially small error.

An interesting feature of this result is that, as in the definition of ${\sf QMA}$, it is not clear that a completeness which is doubly-exponentially close to 1 is independent of the choice of gate set; using the Solovay-Kitaev theorem to convert gate sets potentially gives too large approximation errors.
Indeed, we prove \Cref{thm:QMA-small-completeness-error} with some constraints on the choice of gate set.
This dependence is, however, relatively mild, potentially shedding light on the gate-set dependence of the definition of ${\sf QMA}_1$ as well.
We discuss this more in \Cref{sec:amplification}. For a comparison between the resources used by our ${\sf QMA}$ amplifier, and existing ones, refer to~\Cref{tab:comparison}.

\paragraph{Techniques.} To show that ${\sf QMA}\subseteq {\sf QMA}_1^\infty$, we show how to ``purify'' any ${\sf QMA}$ verifier to get a verifier that has perfect completeness, but uses an infinite dimensional register, as described above. We give two similar, but slightly different constructions (\Cref{sec:first-construction} and \Cref{sec:second-construction}), both based on~\cite{belovs2024purifier}, which shows how to map any bounded-error quantum algorithm to a perfect \emph{transducer}, which is an abstraction of a quantum algorithm. A transducer, even a perfect one, when implemented as a quantum algorithm, will have bounded error. However, perfect transducers (or transducers that are sufficiently close to perfect) can be composed without incurring log factors, and then the composed transducer can be implemented as a bounded-error quantum algorithm, leading to composition of bounded-error quantum algorithms without log factors.
Such a purifier already appeared in~\cite{belovs2024transducers}, but the construction in~\cite{belovs2024purifier} is simpler and more efficient, lending it well to the new setting of \emph{purifying $\sf QMA$ witnesses}.
In our constructions, the witness is infinite dimensional, but with exponentially decreasing amplitudes, conceptually not unlike a history state of a Las Vegas algorithm that has non-zero, but exponentially decreasing probability of making it to any time step $d\in\mathbb{N}$.
We believe this technique may have further applications in quantum complexity theory.

\section{Preliminaries}\label{sec:prelim}

We let upper case letters, such as $A$, denote quantum registers, which are associated with some Hilbert space $\HH_A$. We use the shorthand $\ket{\psi}\in A$ to indicate that $\ket{\psi}\in \HH_A$ is a normalized quantum state. We assume every Hilbert space contains some ``trivial'' state $\ket{0}$.

A (promise) problem is a pair of disjoint subsets $(L_{\text{yes}},L_{\text{no}})$ of $\{0,1\}^*$. The strings in $L_{\text{yes}}$ are the ``yes instances'', and the strings in $L_{\text{no}}$ are the ``no instances''. Some strings might be neither (which is what makes $L$ a promise problem).

A \emph{uniformly-generated polynomial-time quantum circuit} is a mapping from strings $x$ to quantum circuits $V(x)$, specified by a list of gates $U_1,\dots,U_t$, such that there is a polynomial-time Turing machine that outputs $V(x)$ on input $x$. Note that this implies that $t$ is polynomial in the length of $x$.

A \emph{verifier} is a circuit $V$ that takes as input a \emph{witness system} $W$, that will contain the witness, and an \emph{ancilla system} $A$ that will be initialized in a fixed state $\ket{0}$.
The output of $V$ consists of the first qubit, which we denote by $M$, and which is measured in the standard basis, and a remainder we denote by $N$.
We say the circuit \emph{accepts} (outputs ``yes'') if we measure a 1 on $M$, and otherwise, we say it \emph{rejects} (outputs ``no''). We let $\Pi_{\text{acc}}=\ket{1}\bra{1}_M \otimes I_N$ be the orthogonal projector onto states with $\ket{1}$ in the first qubit. This situation is summarized in the following diagram:
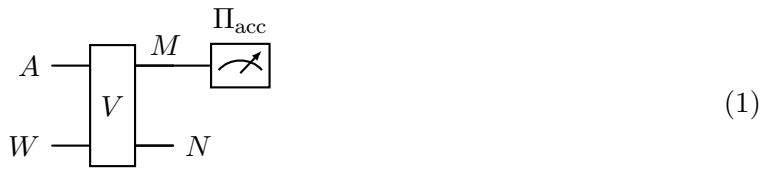
\begin{equation}\label{eq:verification circuit visual}
    \begin{quantikz}
        \lstick{$A$} & \gate[2]{V} & \qw["$M$"{above,pos=0.2}]{} & \meter{\Pi_{\accept}}\\
        \lstick{$W$} &  & \qw \rstick{$N$}
    \end{quantikz}
\end{equation}
Given a witness state $\ket{w}$, we denote by
\begin{align*}
    p_{\accept}(w) = \norm{\Pi_{\mathrm{acc}}V(x)\ket{0}_A\ket{w}_W}^2
\end{align*}
the acceptance probability when using the witness $\ket{w}$.
We now formally define the complexity class ${\sf QMA}$.

\begin{definition}[${\sf QMA}$ and ${\sf QMA}_1$]\label{def:QMA}
    For $s,c\in [0,1]$ such that $c>s$, ${\sf QMA}(c,s)$ is the set of promise problems $L=(L_{\mathrm{yes}},L_{\mathrm{no}})$ such that there exists a uniformly-generated polynomial-time quantum circuit $V(x)$ on registers $A$ and $W$ satisfying:
    \begin{description}
        \item[Completeness:] If $x\in L_{\mathrm{yes}}$, there exists a state $\ket{w} \in W$, called a \emph{witness}, such that
              $$p_{\accept}(w) \geq c$$ (i.e. the verifier accepts with probability at least $c$).
        \item[Soundness:] If $x\in L_{\mathrm{no}}$, for all states
              $\ket{w}\in W$,
              $$p_{\accept}(w) \leq s$$
              (i.e. no witness can convince the verifier to accept with probability more than $s$).
    \end{description}
    The class ${\sf QMA}$ is defined to be ${\sf QMA}(2/3,1/3)$, and the class ${\sf QMA}_1$ is defined to be ${\sf QMA}(1,1/3)$.
\end{definition}

It is known that for any polynomial $r$~\cite{kitaev2002classical,marriott2005QAMGames},
$${\sf QMA}={\sf QMA}\left(1-{2^{-r}}, {2^{-r}}\right)={\sf QMA}\left(\frac{1}{2}+\frac{1}{r}, \frac{1}{2}-\frac{1}{r}\right)
    \quad\mbox{and}\quad
    {\sf QMA}_1=
    {\sf QMA}\left(1,1-\frac{1}{r}\right).$$

In \Cref{def:QMA}, a subtlety is that the definition of ${\sf QMA}_1$ in principle also depends on the choice of (universal) gate set used by $V$. We suppress this dependence in the notation.

It is necessary to allow the verifier in \Cref{def:QMA} an ancillary register $A$ in addition to its witness register $W$ -- otherwise, there always exists a witness $V(x)^\dagger\ket{1}\ket{\phi}$ (for any $\ket{\phi}\in N$) that is accepted with probability 1.

The verification procedure defines a POVM, where the element corresponding to acceptance is is given by
\begin{align}\label{eq:PV}
    P_V = (\bra{0}_A \ot I_W) V(x)^\dagger \Pi_{\accept} V(x) (\ket{0}_A \ot I_W),
\end{align}
which is such that $p_{\accept}(w) = \Tr[P_V \ket{w}\bra{w}]$.
We write $\Pi_0 = \ket{0}\bra{0}_A \ot I_W$, which projects the ancilla system onto $\ket{0}_A$.
A useful fact, which is the basis of the quantum rewinding technique of \cite{marriott2005QAMGames} for ${\sf QMA}$ amplification is the following.

\begin{lemma}\label{lem:rewinding}
    Let $\ket{w}$ be an eigenvector of $P_V$ with eigenvalue $p$ (so $p_{\accept}(w) = p$).
    Then there exist states
    \begin{align*}
        \begin{array}{ll}
            \ket{w_0}  = \ket{0}_A \ket{w}_W,                         & \ket{w_1}_{AW},                      \\
            \ket{s_0} = \ket{0}_M \ket{\phi_0}_N, \mbox{ and } \qquad & \ket{s_1} = \ket{1}_M \ket{\phi_1}_N
        \end{array}
    \end{align*}
    where $\Pi_0 \ket{w_1} = 0$, which are such that for $V = V(x)$
    \begin{align*}
        V \ket{w_0} & = \sqrt{1 - p} \ket{s_0} + \sqrt{p} \ket{s_1} \\
        V \ket{w_1} & = \sqrt{p}\ket{s_0} - \sqrt{1 - p} \ket{s_1}.
    \end{align*}
\end{lemma}

In other words, \Cref{lem:rewinding} identifies two two-dimensional subspaces, such that $V$, with respect to the bases $\{\ket{w_0}, \ket{w_1}\}$ and $\{\ket{s_0}, \ket{s_1}\}$, acts as
\begin{align*}
    V = \bordermatrix{ & \ket{w_0}  & \ket{w_1} \cr
    \ket{s_0}          & \sqrt{1-p} & \sqrt{p} \cr
    \ket{s_1}          & \sqrt{p}   & -\sqrt{1-p} }.
\end{align*}
This applies in particular to the witness with \emph{maximal} acceptance probability, which is an eigenvector of $P_V$.

\begin{proof}
    For completeness, we translate the proof from~\cite{marriott2005QAMGames} to our notation.
    Since the acceptance probability is $\bra{w_0}V^\dagger\Pi_{\text{acc}}V\ket{w_0}=p$, we have $V\ket{w_0}=\sqrt{1-p}\ket{s_0}+\sqrt{p}\ket{s_1}$ for some states of the form $\ket{s_0}=\ket{0}\ket{\phi_0}$ and $\ket{s_1}=\ket{1}\ket{\phi_1}$.
    Define
    $$\ket{w_1}:=-\frac{1}{\sqrt{1-p}}(I-\Pi_0)V^\dagger\ket{s_1}=-\frac{1}{\sqrt{p(1-p)}}(I-\Pi_0)V^\dagger \Pi_{\text{acc}}V\ket{w_0},$$
    which clearly satisfies $\Pi_0\ket{w_1}=0$. Moreover, using $\Pi_0V^\dagger\Pi_{\text{acc}}V\ket{w_0}=p\ket{w_0}$, we have
    \begin{align*}
        \norm{\ket{w_1}}^2 & = \frac{1-\norm{\Pi_0V^\dagger\frac{1}{\sqrt{p}}\Pi_{\text{acc}}V\ket{w_0}}^2}{1-p}=\frac{1-\frac{p^2}{p}}{1-p}=1.
    \end{align*}
    We have:
    \begin{align*}
        \bra{s_1}V\ket{w_1} & = -\frac{\bra{w_0}V^\dagger \Pi_{\text{acc}}}{\sqrt{p}}V\frac{(I-\Pi_0)V^\dagger \Pi_{\text{acc}}V\ket{w_0}}{\sqrt{p(1-p)}} \\
                            & =-\frac{\bra{w_0}V^\dagger\Pi_{\text{acc}}V\ket{w_0}-p\bra{w_0}V^\dagger\Pi_{\text{acc}}V\ket{w_0}}{p\sqrt{1-p}}
        = -\frac{p-p^2}{p\sqrt{1-p}}=-\sqrt{1-p},
    \end{align*}
    and similarly, using $\bra{w_0}(I-\Pi_0)=0$, we have
    \begin{align*}
        \bra{s_0}V\ket{w_1} & = -\frac{\bra{w_0}V^\dagger(I- \Pi_{\text{acc}})}{\sqrt{1-p}}V\frac{(I-\Pi_0)V^\dagger \Pi_{\text{acc}}V\ket{w_0}}{\sqrt{p(1-p)}} \\
                            & =-\frac{\bra{w_0}V^\dagger V(I-\Pi_0)V^\dagger\Pi_{\text{acc}}V\ket{w_0}-(p-p^2)}{(1-p)\sqrt{p}}
        = \frac{p-p^2}{(1-p)\sqrt{p}}=\sqrt{p},
    \end{align*}
    establishing that $V\ket{w_1}=\sqrt{p}\ket{s_0}-\sqrt{1-p}\ket{s_1}$, since $V\ket{w_1}$ is a unit vector.
\end{proof}

\subsection{Bounding errors}
We measure distances between quantum states using the trace distance, defined as
\begin{align*}
    T(\rho,\sigma) = \frac12 \norm{\rho - \sigma}_1
\end{align*}
where $\norm{M}_1 = \Tr[\sqrt{M^\dagger M}]$ is the Schatten 1-norm.
For pure states, $T(\phi,\psi)^2 = 1 - \abs{\braket{\phi}{\psi}}^2$.
We will use the following simple lemma.

\begin{lemma}\label{lem:removing controlled gate}
    Suppose
    $\ket{\psi} = \sqrt{1 - p} \ket{0}\ket{\phi_0} + \sqrt{p} \ket{1}\ket{\phi_1}$
    is a quantum state on a qubit and an additional register $A$, $U$ is a unitary on the register $A$ and
    \begin{align*}
        C[U] := \proj{0} \ot I + \proj{1} \ot U
    \end{align*}
    is a controlled version of $U$.
    Then if $\ket{\phi} = C[U] \ket{\psi}$, we have
    \begin{align*}
        T(\phi, \psi) \leq 2\sqrt{p(1-p)}.
    \end{align*}
\end{lemma}

\begin{proof}
    We bound the overlap
    \begin{align*}
        \abs{\braket{\phi}{\psi}} = \abs{1 - p + p \bra{\phi_1} U \ket{\phi_1}} \geq 1 - 2p
    \end{align*}
    and hence
    \begin{align*}
        T(\phi,\psi) & = \sqrt{1 - \abs{\braket{\phi}{\psi}}^2} \leq \sqrt{1 - (1 - 2p)^2} = 2\sqrt{p(1-p)}. \qedhere
    \end{align*}
\end{proof}

\section{\texorpdfstring{$\sf QMA$}{QMA} with an infinite counter}\label{sec:infinite operations}
We now define a variant of ${\sf QMA}$, which we call ${\sf QMA}^{\infty}$, where the prover and verifier are allowed to use a quantum register representing an ``infinite counter''.
It has a basis of states $\ket{d}$ for $d \in \mathbb{N}$ and the verifier can apply (controlled) increments and check whether the register is in the state $\ket{0}$. 
Clearly, when allowing infinite-dimensional systems in quantum complexity classes, one has to be careful what operations are allowed, since this could significantly impact the computational power.
We will shortly give a precise definition of ${\sf QMA}^{\infty}$. In \Cref{sec:proof truncating witness}, we will prove that in fact ${\sf QMA}^{\infty} = {\sf QMA}$ (\Cref{thm:QMA-inf}), so allowing the prover and verifier an infinite counter is not especially powerful.
What makes this addition interesting, however, is that one does get perfect completeness: ${\sf QMA}^{\infty} = {\sf QMA}_1^{\infty}$ (\Cref{thm:QMAisQMA1-inf}), as we will prove in the next section.

Consider a universal gate set $\mathcal C_n$ on $n$ qubits.
We now define a new gate set $\mathcal C_n^{\infty}$ which acts on $n$ qubits, as well as on one additional infinite dimensional quantum system $B$ which has a basis $\ket{d}$ for $d \in \mathbb{N}$.
The new gate set $\mathcal C_n^{\infty}$ consists of the gates in $\mathcal C_n$ (on the original $n$ qubits), as well as the following operations:
\begin{description}
    \item[Controlled increments:] We let $\incr$ denote the isometry on $B$ that maps $\ket{d} \mapsto \ket{d+1}$. We allow the controlled increment operation between qubit $i$ and $B$:
          \begin{align*}
              \proj{0}_i \ot I_B + \proj{1}_i \ot \incr_B.
          \end{align*}
    \item[Checking the zero state:]
          We allow a gate that checks if $\ket{d}_B=\ket{0}_B$, by applying a Pauli $X$ to any qubit $i$, conditioned on $\ket{0}_B:$
          \begin{align*}
              X_i \ot \proj{0}_B + I \ot \left(\sum_{d \geq 1} \proj{d}_B \right)
          \end{align*}
          where $X_i$ denotes the Pauli $X$ applied to the $i$-th qubit. Note that this allows us to implement any polynomial-sized circuit in ${\cal C}_n$ controlled on whether $d>0$ or $d=0$: we simply apply a conditional Pauli $X$ to a qubit register initialized to $\ket{0}$, and then implement the desired circuit controlled on this register.
\end{description}

Now ${\sf QMA}^\infty(c,s)$ is defined as in \Cref{def:QMA}, but the verifier $V$ can use any polynomial number of gates from the set $\mathcal C_n^{\infty}$,\footnote{In fact, the ${\sf QMA}^\infty$ verifiers we construct need something potentially weaker than access to these two operations: a single call to an operation that does a controlled-increment, and then immediately measures whether $\ket{d}_B=\ket{0}_B$.
} and in particular, the witness system may include the register $B$. Let ${\sf QMA}_1^\infty={\sf QMA}^\infty(1,1/3)$, and ${\sf QMA}^\infty = {\sf QMA}^\infty(2/3,1/3)$. Again, ${\sf QMA}_1^{\infty}$ in principle also depends on the gate set $\mathcal C_n$; we will show later that this gate set dependence is very mild.

We note that if one considers polynomial quantum computations with an initial state where all qubits as well as the infinite counter are in the $\ket{0}$ state, leading to a class ${\sf BQP}^\infty$, it is immediate that the computation can be truncated to a polynomial number of qubits. Indeed, the number of times $\ell$ that a (controlled) increment is applied is at most polynomial, and the counter $B$ will never reach states $\ket{d}$ for $d > \ell$, so $B$ can be truncated to $\log(\ell)$ qubits. That is, ${\sf BQP}^\infty$ is directly seen to equal ${\sf BQP}$. This is even true of $\sf RQP$ and $\sf EQP$, the one-sided and exact variants of ${\sf BQP}$.
In \Cref{sec:proof truncating witness}, we show that for ${\sf QMA}^\infty$ we can similarly truncate the infinite counter: ${\sf QMA}={\sf QMA}^\infty$ (\Cref{thm:QMA-inf}).
This result is not immediately obvious, since the witness states can be supported on infinitely many computational basis states $\ket{d}_B$.
The main ingredient is to prove that any witness state can be replaced by a witness state on a low-dimensional subspace without changing the acceptance probability too much.
From \Cref{thm:QMA-inf} we conclude that it not unreasonable to allow access to infinite dimensional registers with appropriately defined operations for proof systems. On the other hand, we will also prove that ${\sf QMA}_1^\infty = {\sf QMA}^\infty = {\sf QMA}$, which is not known without an infinite counter.

A completely different motivation for considering quantum complexity classes with infinite dimensional registers is that many quantum mechanical systems, in particular bosonic systems, are modeled to be inherently infinite dimensional. We comment on this connection in \Cref{sec:bosonic quantum computers}.

The definition we give of $\sf{QMA }^{\infty}$ is a minimal one. We could also have defined ${\sf QMA}^\infty$ with multiple counters or with an invertible two-sided infinite counter; this also does not increase the computational power of ${\sf QMA}^{\infty}$. We discuss this in \Cref{rmk:extensions}.

A disadvantage of our minimal definition of ${\sf QMA}_1^\infty$ is that it is not immediately clear how to do amplification (error reduction). We cannot directly apply rewinding as in~\cite{marriott2005QAMGames}, since the increment operation is not invertible, and we cannot use parallel repetition, since there is only a single infinite counter register.

Thus, it is not immediately clear that ${\sf QMA}^\infty(1,s)={\sf QMA}_1^\infty$ for $s\neq 1/3$, even for constant $s$.
However, our results nonetheless imply the following.
\begin{theorem}\label{thm:QMA-inf-with-different-c-s}
    For any polynomial $r$,
    $${\sf QMA}_1^\infty={\sf QMA}^\infty\left(1,1-\frac{1}{r}\right)={\sf QMA}^\infty(1,2^{-r})$$
    $$\mbox{and}\quad{\sf QMA}^\infty={\sf QMA}^\infty(1-2^{-r},2^{-r})={\sf QMA}^\infty\left(\frac{1}{2}+\frac{1}{r},\frac{1}{2}-\frac{1}{r}\right).$$
\end{theorem}
\noindent This follows from the fact that all classes mentioned in the theorem are in fact equal to ${\sf QMA}$, which is proven in \Cref{cor:QMA1-inf-is-QMA-inf-is-QMA}.

\section{Perfect completeness of \texorpdfstring{$\sf QMA^\infty$}{QMA with an infinite register}}\label{sec:main}

In this section, our goal is to prove that ${\sf QMA}\subseteq{\sf QMA}_1^\infty$.
To that end, we need to show that if we have a problem in ${\sf QMA}$ with verification circuit $V(x)$, we can use it to construct a ${\sf QMA}_1^\infty$ verification circuit $V'$ -- on the one hand, we have the stringent requirement that $V'$ have perfect completeness (while keeping soundness sufficiently small), on the other hand, it may use an infinite-dimensional register equipped with the operations described in~\Cref{sec:infinite operations}.

We present two constructions of this form. The first, in \Cref{sec:first-construction} has a conceptually simple analysis, and uses only one call to each of $V$ (the original verifier) and $V^\dagger$. It has the disadvantage that it requires a constant gap between the completeness and soundness of $V$, and results in a soundness $>1/2$, so it only implies the weaker statement ${\sf QMA}\subseteq{\sf QMA}^\infty(1,1/2+2^{-r})$ (recall that error reduction is not immediate in our infinite-dimensional model).
The second construction, in \Cref{sec:second-construction}, is similar in spirit, but works for any soundness-completeness gap, centered around $1/4$, increasing the soundness by at most a constant factor, and thus establishing ${\sf QMA}\subseteq {\sf QMA}_1^\infty$ (\Cref{cor:QMAinQMA-inf-exp-soundness}). Compared with the first construction, it uses one additional call to $V$.

\subsection{A first construction, using only 2 calls to \texorpdfstring{$V$ and $V^\dagger$}{the verifier}}\label{sec:first-construction}

Fix a ${\sf QMA}(c,s)$ verification circuit $V(x)$ for a language $L$, as in \cref{eq:verification circuit visual}, acting on a register $A$ initialized in the state $\ket{0}_A$ and a register $W$ containing the witness state.
Recall that $\Pi_{\accept}$ denotes the projection operator on the accepting states and $\Pi_0$ denotes the projection onto the subspace where the ancilla qubits of the verifier are in the state $\ket{0}_A$.
Consider the following procedure, which we will refer to as $V'$:
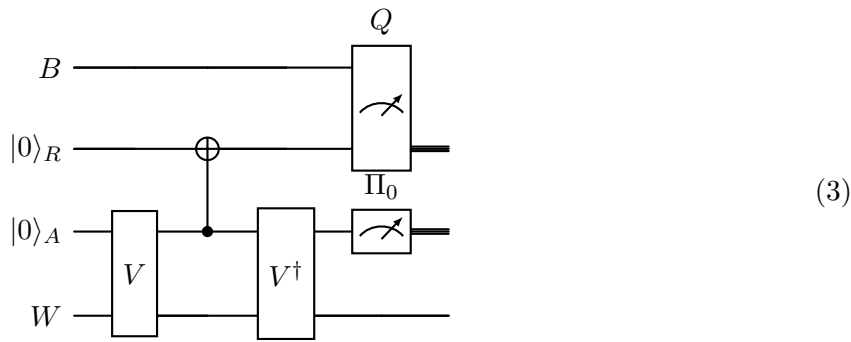
\begin{equation}\label{eq:modified-verifier-1}
    \begin{quantikz}
        \lstick{$B$} & \qw & \qw & \qw & \meter[2]{Q} \\
        \lstick{$\ket{0}_R$} & \qw & \targ{} & \qw &  & \cw\\
        \lstick{$\ket{0}_A$} & \gate[2]{V} & \ctrl{-1} & \gate[2]{V^{\dagger}} & \meter{\Pi_0} & \cw\\
        \lstick{$W$} &  & \qw & & \qw & \qw
    \end{quantikz}
\end{equation}
The circuit $V'$ takes a witness state on the tensor product of the original witness register $W$, and a single infinite counter $B$ with basis $\{\ket{d}:d\in\mathbb{N}\}$. Its ancilla consists of $A$, the ancilla of $V$, as well as a single-qubit register $R$.
The measurement $\{Q, I-Q\}$ on $BR$ is defined by letting $Q$, associated to measurement outcome $q = 1$, be the projection operator onto the span of the states
\begin{align}\label{eq:accepting measurement}
    \{\ket{0}\ket{1} \} \cup \{\ket{\Phi_k} := \frac{1}{\sqrt2}(\ket{d}\ket{0} + \ket{d+1}\ket{1})\}_{d=0}^{\infty}.
\end{align}

We also measure $\{\Pi_0, I - \Pi_0\}$ on the register $A$ with respective outcomes $a = 1, 0$.
The verifier $V'(x)$ accepts if $a = 0$ (so the ancilla is not in the state $\ket{0}_A$), or if $q = 1$.

The measurement $\{Q, I-Q\}$ on systems $B$ and $R$ can be implemented using the controlled increment operator $\incr$ and the zero-state measurement  on the infinite counter as defined in \Cref{sec:infinite operations}:
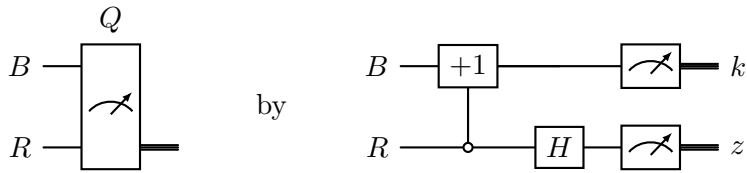
\begin{center}
    \begin{quantikz}
        \lstick{$B$} \qw & \meter[2]{Q} \\
        \lstick{$R$} \qw & & \cw
    \end{quantikz}
    \qquad by \qquad
    \begin{quantikz}
        \lstick{$B$} \qw & \gate{\circincr} & \qw & \meter{} & \cw \rstick{$k$}\\
        \lstick{$R$} \qw & \octrl{-1} & \gate{H} & \meter{} & \cw \rstick{$z$}
    \end{quantikz}
\end{center}
where the measurement $Q$ has outcome $q = 1$ if $k=0$ or $z=0$.
This means that $V'(x)$ is a legitimate verification procedure for the class ${\sf QMA}_1^\infty$.

\begin{proposition}\label{prop:perfect completeness}
    There exists a constant $s_0 > 0$, such that for any problem in ${\sf QMA}$ with verification circuits $V(x)$ with completeness $c > \frac12$ and soundness $s < s_0$, the verifier $V'(x)$ described above has perfect completeness $c' = 1$ and soundness $s' \leq \frac12 + 2\sqrt{s(1 - s)} < 1$.
\end{proposition}

\begin{proof} We prove completeness first, and then soundness.
    \paragraph{Completeness.} If $x \in L_{\text{yes}}$, there exists a witness with maximal acceptance probability $p \geq c$ for $V(x)$. This witness $\ket{w}$ is an eigenvector of the operator $P_V$.
    We consider the following state as witness for $V'(x)$:
    \begin{align*}
        \ket{\psi}_B\ket{w}_W \quad \text{ where } \quad \ket{\psi}_B = \sqrt{1 - \gamma^2} \sum_{d=0}^{\infty} \gamma^d \ket{d} \quad \text{ and } \quad \gamma = \frac{1 - p}{p}.
    \end{align*}
    Note that by assumption, $\gamma \leq \frac{1}{c}-1< 1$ (which is necessary for $\ket{\psi}_B$ to be well-defined).
    A direct consequence of \Cref{lem:rewinding} is that the circuit
    
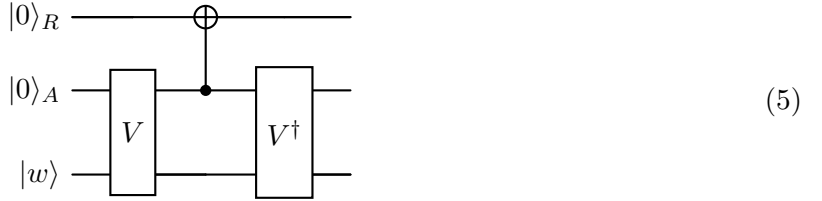
\begin{equation}\label{eq:distill}
        \begin{quantikz}
            \lstick{$\ket{0}_R$} & \qw & \targ{} & \qw & \qw \\
            \lstick{$\ket{0}_A$} & \gate[2]{V} & \ctrl{-1} & \gate[2]{V^{\dagger}} & \\
            \lstick{$\ket{w}$} &  & \qw & & \qw
        \end{quantikz}
    \end{equation}
    leads to the state
    \begin{align*}
        p \ket{1}_R\ket{0}_A \ket{w}_W + 2\sqrt{p(1-p)}\ket{\phi}_{RAW} + (1 - p) \ket{0}_R\ket{0}_A\ket{w}_W
    \end{align*}
    where $\ket{\phi}_{RAW}$ satisfies $I \ot \Pi_0 \ket{\phi}_{RAW} = 0$.
    If the measurement outcome on $A$ is $a = 0$, we accept (which does not compromise perfect complenetess).
    If, on the other hand, $a = 1$, so the ancilla register $A$ is measured to be in the state $\ket{0}_A$, the post-measurement state is given by
    \begin{align}\label{eq:q given p}
        (\sqrt{q} \ket{1}_R + \sqrt{1 - q} \ket{0}_R) \ket{0}_A\ket{w}_W,
        \quad\mbox{where}\quad
        q = \frac{p^2}{p^2 + (1-p)^2} > \frac12 \, .
    \end{align}
    In particular, on the register $R$ we are left with the state $\sqrt{q} \ket{1} + \sqrt{1 - q} \ket{0}$.
    We can now rewrite the state on $BR$ before measuring $Q$ as
    \begin{align*}
         & \ket{\psi}_B (\sqrt{q}\ket{1} + \sqrt{1-q}\ket{0})                                                                                                                      \\
         & \qquad = \sqrt{q(1 - \gamma^2)} \ket{0}\ket{1} + \sqrt{1 - \gamma^2} \sum_{d=0}^{\infty} (\gamma^d \sqrt{1-q} \ket{d}\ket{0} + \gamma^{d+1}\sqrt{q} \ket{d + 1}\ket{1}) \\
         & \qquad = \sqrt{q(1 - \gamma^2)} \ket{0}\ket{1} + \sqrt{1 - \gamma^2} \sum_{d=0}^{\infty} \gamma^d \sqrt{1- q} (\ket{d}\ket{0} +  \ket{d+1}\ket{1})
    \end{align*}
    using that $\gamma \sqrt{q} = \sqrt{1 - q}$.
    This is a superposition of the states defined in \cref{eq:accepting measurement}, so we conclude that using the witness $\ket{\psi}_B \ket{w}_W$, the verifier accepts with probability $1$.
    \paragraph{Soundness.} We consider $x \in L_{\text{no}}$ and an arbitrary witness state $\ket{w}_{BW}$. Our goal is to bound the acceptance probability $p_{\accept}(w)$ for $V'(x)$.
    By assumption, if $\ket{\phi}$ denotes the state after the first application of $V$, $\ket{\phi}$ has bounded amplitude on the accepting state: $\norm{\Pi_{\accept}\ket{\phi}} \leq s$.
    That means that by \Cref{lem:removing controlled gate}, \emph{not} applying the CNOT in the verification circuit gives a trace distance error at most $2\sqrt{s(1-s)}$.
    Let $\tilde{V}(x)$ denote a modification of $V'(x)$ without the CNOT gate; the trace distance bound implies that if $\tilde p_{\accept}(w)$ denotes the acceptance probability of $\tilde V(x)$ on the witness $\ket{w}$, we have
    \begin{align*}
        p_{\accept}(w) \leq \tilde p_{\accept}(w) + 2\sqrt{s(1-s)}.
    \end{align*}
    Now, $\tilde V(x)$ simplifies as
    
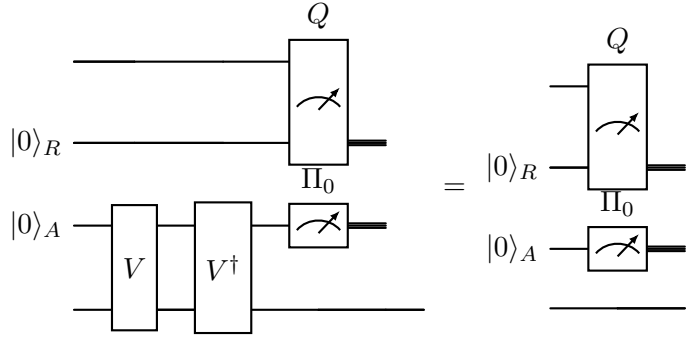
\begin{center}
        \begin{quantikz}
            & \qw & & \meter[2]{Q} \\
            \lstick{$\ket{0}_R$} & \qw  & \qw & & \cw\\
            \lstick{$\ket{0}_A$} & \gate[2]{V} & \gate[2]{V^{\dagger}} & \meter{\Pi_0} & \cw\\
            &  & \qw & & \qw & \qw
        \end{quantikz}
        =
        \begin{quantikz}
            \qw &\meter[2]{Q} \\
            \lstick{$\ket{0}_R$} & \qw & \cw\\
            \lstick{$\ket{0}_A$} &\meter{\Pi_0} & \cw \\
            &  & \qw
        \end{quantikz}
    \end{center}
    The measurement on $A$ measures the state $\ket{0}$ with probability 1. If $w_B$ denotes the reduced state of the witness on $B$, the state before the final measurement is $\rho_{BR} = w_B \ot \proj{0}_R$ and
    \begin{align*}
        \tilde p_{\accept}(w) = \Tr[Q \rho_{BR}] = \sum_{k=0}^{\infty} \bra{\Phi_k} \rho_{BR} \ket{\Phi_k} = \frac12
    \end{align*}
    using that $\rho_{R} = \proj{0}_R$.
    We conclude that
    \begin{align*}
        p_{\accept}(w) \leq \frac12 + 2\sqrt{s(1-s)} = s', 
    \end{align*}
    which is certainly less than 1 if $s<s_0=\frac{1}{16}$.
\end{proof}

\begin{theorem}\label{thm:qma1 with infinite tape}
    Choose a gate set with exact inverses, and exact CNOT and Hadamard gates. Then ${\sf QMA}\subseteq{\sf QMA}^\infty(1,\frac{1}{2}+2^{-r})$.
\end{theorem}

\begin{proof}
    By \Cref{prop:perfect completeness}, ${\sf QMA}(c, s) \subseteq {\sf QMA}^\infty(1,\frac{1}{2}+2\sqrt{s(1-s)})$ for $c > \frac12$ and $s < s_0$. Taking (for example) $c=2/3$ and $s=2^{-2(r+1)}$, we have $2\sqrt{s(1-s)}\leq 2\cdot 2^{-(r+1)}=2^{-r}$, so
    ${\sf QMA}={\sf QMA}(c,s)\subseteq{\sf QMA}^\infty(1,\frac{1}{2}+2^{-r})$.
    The conditions on the gate set are required to perform the verification in \cref{eq:modified-verifier-1} without approximation error.
\end{proof}

\begin{remark}
    The technique here cannot be used to boost ${\sf BQP}$ algorithms to have zero error, or even ${\sf BQP}^\infty$ algorithms. However, a similar construction can be used to turn such algorithms into a kind of error-free object called a \emph{transducer}~\cite{belovs2024purifier}. Such objects can be composed without the need for further boosting, and then truncated back down to bounded-error quantum algorithms, leading to composition of ${\sf BQP}$ algorithms without log factors.
\end{remark}

\subsection{A second construction, with improved soundness}\label{sec:second-construction}

We now give an alternative construction for showing ${\sf QMA} \subseteq {\sf QMA}_1^\infty$.
It uses one additional call to the original verification circuit $V(x)$, but has the advantage that it does not require the initial verifier to have a constant promise gap.
With the same set-up as before, we now let $V'(x)$ be the following verifier:
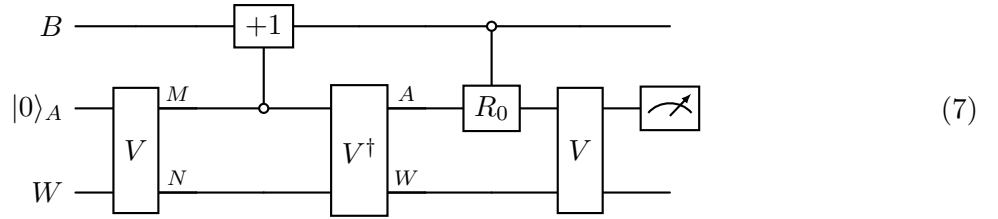
\begin{equation}\label{eq:modified qma 2}
    \begin{quantikz}
        \lstick{$B$} & & & \gate{\circincr} & & & \octrl{1} && \\
        \lstick{$\ket{0}_A$} & \gate[2]{V} & \wire[l][1]["M"{above,pos=.5}]{a} & \octrl{-1} &
        \gate[2]{V^\dagger} & \wire[l][1]["A"{above,pos=.5}]{a} & \gate{R_0} & \gate[2]{V} & \meter{}\\
        \lstick{$W$} && \wire[l][1]["N"{above,pos=.5}]{a} &&&\wire[l][1]["W"{above,pos=.5}]{a}&& &
    \end{quantikz}
\end{equation}
where we recall the open control on $B$ means that $R_0=2\proj{0}_A-I_A$ is performed controlled on $\ket{d}_B$ with $d \geq 1$.
Again, this defines a valid verification circuit for ${\sf QMA}_1^\infty$ (see \Cref{sec:infinite operations}), although we have not yet established its soundness and completeness.

\begin{theorem}\label{thm:perfect completeness 2}
    For any problem in ${\sf QMA}$ with verification circuits $V(x)$ with completeness $c > \frac14$ and soundness $s < \frac14$, the verifier $V'(x)$ described above has perfect completeness $c' = 1$ and soundness $s' \leq 1-(1 - s)(1 - 4s)^2 = 9s-24s^2+16s^3 < 1$.
\end{theorem}
\begin{proof}
    We first compute what happens when the witness state is of the form
    \begin{align}\label{eq:product witness}
        \ket{\psi}_B \ket{w}_W, \qquad \ket{\psi}_B = \sum_{d = 0}^{\infty} \psi_d \ket{d},
    \end{align}
    where $\ket{w}$ is an eigenvector of $P_V$ (see \cref{eq:PV}) with acceptance probability $p$ for $V(x)$.
    After applying $V$, we have a state of the form
    \begin{align*}
        \ket{\psi}_B(\sqrt{1 - p} \ket{0}\ket{\phi_0} + \sqrt{p}\ket{1}\ket{\phi_1}),
    \end{align*}
    which the controlled incrementing operation maps to a state
    \begin{align}\label{eq:BMN}
        \ket{\Psi}_{BMN} = \psi_0 \sqrt{p} \ket{0}_B \ket{1} \ket{\phi_1} + \sum_{d=1}^{\infty} \ket{d}_B( \psi_{d-1} \sqrt{1-p}\ket{0}\ket{\phi_0} + \psi_{d} \sqrt{p} \ket{1}\ket{\phi_1}).
    \end{align}
    Note that $R_0\ket{w_0}=\ket{w_0}$ and $R_0\ket{w_1}=-\ket{w_1}$, meaning $R_0$ acts as a Pauli $Z$ gate on the states $\ket{w_0}$ and $\ket{w_1}$.
    Thus, by \Cref{lem:rewinding}, $V R_0 V^\dagger$ acts as the reflection operator
    \begin{align*}
        V R_0 V^\dagger & = \begin{pmatrix}
                                \sqrt{1-p} & \sqrt{p} \\ \sqrt{p} & -\sqrt{1-p}
                            \end{pmatrix}
        \begin{pmatrix}
            1 & 0 \\ 0 & -1
        \end{pmatrix}
        \begin{pmatrix}
            \sqrt{1-p} & \sqrt{p} \\ \sqrt{p} & -\sqrt{1-p}
        \end{pmatrix}
        =
        \begin{pmatrix}
            1 - 2p & 2\sqrt{p(1-p)} \\ 2\sqrt{p(1-p)} & 2p - 1
        \end{pmatrix}
    \end{align*}
    on the basis states $\ket{s_0}=\ket{0}\ket{\phi_0}$ and $\ket{s_1}=\ket{1}\ket{\phi_1}$.
    The verifier $V'(x)$ applies this operation controlled on $d \geq 1$.
    It is now easy to compute the \emph{rejection} probability (i.e. the squared amplitude on $\ket{0}_M$ after applying $VR_0V^\dagger$ to the second term of \cref{eq:BMN}) as
    \begin{equation}\label{eq:rejection probability modified verifier}
        \begin{split}
            p_{\text{reject}} & = \sum_{d = 1}^{\infty}\norm{(\bra{0}_M\otimes I_N)VR_0V^\dagger (\psi_{d-1} \sqrt{1-p}\ket{0}\ket{\phi_0} + \psi_{d} \sqrt{p} \ket{1}\ket{\phi_1})}^2 \\
                              & =\sum_{d = 1}^{\infty}\abs{(1-2p)\psi_{d-1}\sqrt{1-p}+2\sqrt{p(1-p)}\psi_d\sqrt{p}}^2                                                                  \\
                              & = (1 - p) \sum_{d = 1}^{\infty} \abs{\psi_{d-1}(1 - 2p) + 2\psi_d p}^2.
        \end{split}
    \end{equation}

    \paragraph{Completeness.} For $x \in L_{\text{yes}}$, let $\ket{w}_W$ be the state with maximal acceptance probability $p > c > \frac{1}{4}$ for $V(x)$.
    As witness for $V'(x)$ we take the product state $\ket{\psi}_B \ket{w}_W$ with
    \begin{align}\label{eq:inf-witness}
        \ket{\psi}_B = \sqrt{1 - \gamma^2} \sum_{d=0}^{\infty} \gamma^d \ket{d}, \qquad \gamma = 1 - \frac{1}{2p}.
    \end{align}
    Note that $p > \frac{1}{4}$ implies $\abs{\gamma} < 1$.
    From \cref{eq:rejection probability modified verifier}, and our choice of $\ket{\psi}_B$ we conclude that
    \begin{align*}
        p_{\text{reject}} = (1 - \gamma^2)(1 - p)\sum_{d = 0}^{\infty} \gamma^{2d - 2} \abs{(1 - 2p) + 2\gamma p}^2 = 0,
    \end{align*}
    where we have used that $\gamma = 1 - \frac{1}{2p}$, so the verifier accepts with certainty.
    \paragraph{Soundness: simple but suboptimal argument.} We start with an easy soundness proof similar to that of \Cref{prop:perfect completeness}; it achieves a worse soundness than claimed and requires a soundness $s < \frac15$ for the original verifier. We nevertheless include the argument because of its simplicity (but the reader may skip this didactic excursion).
    We consider an instance $x \in L_{\text{no}}$, and some witness state $\ket{w}_{BW}$ After the first application of $V(x)$, we have a state $\ket{\phi}$ that satisfies $\norm{\Pi_0 \ket{\phi}}^2 \leq s$, so by \Cref{lem:removing controlled gate}, we can apply $\incr$ \emph{without} the control at error at most $2\sqrt{s(1-s)}$.
    Let $\tilde V(x)$ denote the modified verifier where the application of $\incr$ is not controlled, with acceptance probability $\tilde p_{\accept}(w)$. Then
    \begin{align*}
        p_{\accept} \leq \tilde p_{\accept}(w) + 2\sqrt{s(1-s)}.
    \end{align*}
    Now, $\tilde V(x)$ simplifies as
    
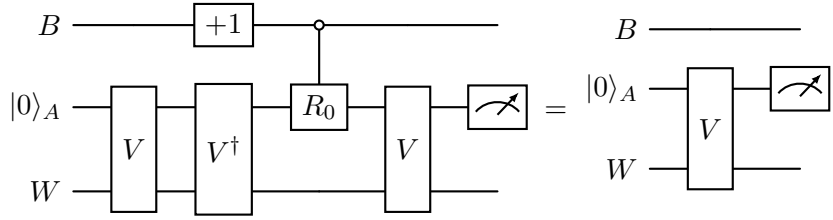
\begin{equation*}
        \begin{quantikz}
            \lstick{$B$} & & \gate{\circincr} & \octrl{1} && \\
            \lstick{$\ket{0}_A$} & \gate[2]{V} &
            \gate[2]{V^\dagger} & \gate{R_0} & \gate[2]{V} & \meter{}\\
            \lstick{$W$} &&&& &
        \end{quantikz}
        =
        \begin{quantikz}
            \lstick{$B$} && \\
            \lstick{$\ket{0}_A$} \qw & \gate[2]{V} & \meter{}\\
            \lstick{$W$}  &&
        \end{quantikz}
    \end{equation*}
    since after applying $\incr$, the state only has support on $\ket{d}$ for $d \geq 1$ and $R_0$ acts trivially on $\ket{0}_A$.
    Thus, by the soundness of $V$, $\tilde p_{\accept}(w) \leq s$.
    We conclude that
    \begin{align*}
        p_{\accept}(w) \leq s + 2\sqrt{s(1-s)} < 1
    \end{align*}
    provided that $s < \frac15$.
    \paragraph{Soundness.} We now give a sharper argument, leading to a better bound for the soundness of $V'(x)$.
    Let $x \in L_{\text{no}}$.
    Consider a basis of eigenvectors $\ket{w^{(k)}}$ for $P_V$ with eigenvalues (i.e. acceptance probabilities) $p_k$.
    We would like to bound the maximal acceptance probability for $V'$, which will correspond to a witness state that is an eigenvector of $P_{V'}$.
    Let
    \begin{align*}
        \HH_k = \HH_B \ot \mathrm{span}\{\ket{w^{(k)}}\}
        \qquad\mbox{so}\qquad
        \HH_{BW} = \bigoplus_k \HH_k.
    \end{align*}
    We claim that $P_{V'}$, which is an operator on $\HH_{BW}$, is block diagonal in these spaces. To prove this, we show that $P_{V'}$ maps each $\HH_k$ to itself.
    Let $\ket{w} = \ket{w^{(k)}}$. Since $\ket{w}$ is an eigenvector of $P_V$, \Cref{lem:rewinding} applies.
    Let ${\cal S}_W\subset \HH_{AW}$ and ${\cal S}_S\subset\HH_{MN}$ denote the subspaces spanned by $\ket{w_0}, \ket{w_1}$ and $\ket{s_0}, \ket{s_1}$ as given in \Cref{lem:rewinding} (note that these depend on $\ket{w}$).
    We can show that $V'$ maps $\HH_B\otimes {\cal S}_W$ to $\HH_B\otimes {\cal S}_B$, and its dual does the reverse, by showing that
    $$\HH_B\ot{\cal S}_W \overset{V}{\leftrightarrow} \HH_B\ot {\cal S}_S
        \overset{C{\incr}}{\leftrightarrow}\HH_B\ot {\cal S}_S
        \overset{V^\dagger}{\leftrightarrow}\HH_B\ot {\cal S}_W
        \overset{C[R_0]}{\leftrightarrow}\HH_B\ot {\cal S}_W
        \overset{V}{\leftrightarrow} \HH_B\ot {\cal S}_S$$
    where the notation $\HH\overset{U}{\leftrightarrow} \HH'$ means the image of $\HH$ under $U$ is in $\HH'$, and the image of $\HH'$ under $U^\dagger$ is in $\HH$; and $C[U]$ denotes a controlled $U$. Each of these statements can be verified by \Cref{lem:rewinding}, or by inspection. Then, since $\ket{0}_A\ket{w}_W\in {\cal S}_W$,
    $$\Pi_{\text{acc}}V'(\HH_B\otimes \ket{0}_A\otimes \ket{w}_W)\in \HH_B\otimes \Pi_{\text{acc}}({\cal S}_S)
        = \HH_B\otimes \mathrm{span}\{\ket{s_1}\}\subset \HH_B\otimes {\cal S}_S$$
    and so
    $${V'}^\dagger\Pi_{\text{acc}}V'(\HH_B\otimes \ket{0}_A\otimes \ket{w}_W)\in \HH_B\otimes {\cal S}_W.$$
    Since $(\bra{0}_A\otimes I_W)({\cal S}_W)=\mathrm{span}\{\ket{w_0}=\ket{0}_A\ket{w}_W\}$ (see \Cref{lem:rewinding}), we have
    $$P_{V'}(\HH_B\ot\ket{w}_W)=(I_B\ot\bra{0}_A\ot I_W){V'}^\dagger\Pi_{\text{acc}}V'(I_B\ot \ket{0}_A\ot I_W)(\HH_B\ot\ket{w}_W)\in \HH_B\otimes \ket{0}_A\ot\ket{w}_W.$$
    We have thus shown that for all $k$, $P_{V'}(\HH_k)=\HH_k$, and so $P_{V'}$ is block diagonal in these spaces.
    In particular, this implies that when maximizing the acceptance probability
    \begin{align*}
        \max_{\ket{\phi} \in \HH_B \ot \HH_W} \Tr[P_{V'} \phi] = \max_{k} \max_{\ket{\phi_k} \in \HH_k} \Tr[P_{V'} \phi_k]
    \end{align*}
    it suffices to consider product state witnesses of the form $\ket{\psi}_B \ket{w}_W$.
    We have already computed the rejection probability of such witnesses in \cref{eq:rejection probability modified verifier}.
    Consider a product witness as in \cref{eq:product witness}, where the acceptance probability of $\ket{w}_W$ for $V(x)$ is $p \leq s$.
    We bound
    \begin{equation}\label{eq:p-reject-soundness}
        \begin{split}
            p_{\text{reject}} & = (1 - p) \sum_{d = 0}^{\infty} \abs{\psi_{d-1}(1 - 2p) + 2\psi_d p}^2                                                                   \\
                              & \geq (1 - p)\sum_{d=1}^{\infty} \left(\abs{\psi_{d-1}}^2 (1 - 2p)^2 + \abs{\psi_{d}}^2 (2p)^2 - 4p(1 - 2p)\abs{\psi_{d-1} \psi_d}\right) \\
                              & \geq (1-p)\left( (1 - 2p)^2 + (2p)^2(1 - \abs{\psi_0}^2) -4p(1 - 2p) \sqrt{1 - \abs{\psi_0}^2} \right)                                   \\
                              & = (1 - p)\left(1 - 2p - 2p\sqrt{1 - \abs{\psi_0}^2}\right)^2                                                                             \\
                              & \geq (1 - p)(1 - 4p)^2 \geq (1 - s)(1 - 4s)^2.
        \end{split}
    \end{equation}
    Here, in the second line we use that for $z_1, z_2 \in \CC$, $$\abs{z_1 + z_2}^2 \geq (\abs{z_1} - \abs{z_2})^2 = \abs{z_1}^2 + \abs{z_2}^2 - 2\abs{z_1 z_2}$$ and $(1 - 2p) > 0$.
    The remainder follows from the Cauchy-Schwartz inequality, and the fact that $\sum_{d \geq 1} \abs{\psi_{d}}^2 = 1 - \abs{\psi_0}^2$.
    Finally, we conclude that since we assumed $s < \frac14$,
    \begin{equation*}
        p_{\accept} = 1 - p_{\text{reject}} \leq 1 - (1 - s)(1 - 4s)^2 < 1.\qedhere
    \end{equation*}
\end{proof}

\begin{corollary}\label{cor:QMAinQMA-inf-exp-soundness}
    Choose a gate set with inverses, in which controlled reflection around $\ket{0}_A$ can be implemented exactly.
    For any polynomial $r$, ${\sf QMA}\subseteq{\sf QMA}^\infty(1,2^{-r})$, and thus, ${\sf QMA}\subseteq {\sf QMA}_1^\infty$.
\end{corollary}
\begin{proof}
    Let $L\in{\sf QMA}$. For any polynomial $r'$, let $V'$ be a ${\sf QMA}(1-2^{-r'},2^{-r'})={\sf QMA}$ verifier for $L$. By \Cref{thm:perfect completeness 2}, there is a ${\sf QMA}^\infty(c',s')$ verifier $V'$ with perfect completeness $c'=1$, and soundness $s'=O(s)=O(2^{-r'})$. If we take $r'$ a sufficiently large polynomial, we have $s'\leq 2^{-r}$, establishing $L\in {\sf QMA}^\infty(1,2^{-r})$.
\end{proof}

Note that this proves in particular that the dependence on the initial choice of gate set is limited; for a broad class of gate sets (for example, Clifford+T) \Cref{cor:QMAinQMA-inf-exp-soundness} 
applies and hence ${\sf QMA}_1^\infty = {\sf QMA}$ is independent of the choice of gate set within this broad class. What happens if, e.g., one does  \emph{not} allow exact inverses is an open question.

\section{{\sf QMA} Amplification} \label{sec:amplification}

A number of past works have shown how to use a $\sf QMA(c,s)$ verifier $V$ to construct a verifier $V'$ with better completeness and soundness~\cite{kitaev1999QuantumNP,marriott2005QAMGames,nagaj2009FastAmpQMA,fefferman2016errorReduction}, considering efficiency in various resources, including:
\begin{itemize}
    \item The number of qubits used by Merlin to communicate the witness: $l_M'=\log\dim W'$.
    \item The number of qubits used by Arthur to run the verifier: $l_A'=\log\dim W' +\log\dim A'$.
    \item The time used by Arthur, that is, the number of elementary gates in the verifier, $t_A'$.
\end{itemize}
Such a procedure is a called a \emph{{\sf QMA} amplifier}. Since the above mentioned ${\sf QMA}$ amplifiers can reduce soundness and completeness error to $2^{-r}$ using resources that scale polynomially in $r$, they imply ${\sf QMA}={\sf QMA}(1-2^{-r},2^{-r})$ for any polynomial $r$.

If we truncate the infinite-dimensional register in the constructions in \Cref{sec:main} we get a standard ${\sf QMA}$ verifier, on a finite space, but which no longer has perfect completeness. It does have significantly less completeness error than the original verifier though, making it a ${\sf QMA}$ amplifier\footnote{The terminology ${\sf QMA}$ amplification historically refers to increasing the \emph{gap} $c-s$. We use a slightly different meaning here, but it is similar in spirit.} (\Cref{thm:truncate}). Unlike the {\sf QMA} amplifiers cited above, the soundness error is not improved (it is slightly worse) but the upside is that we can reduce the completeness error to $2^{-q}$ using resources that scale \emph{logarithmically in $q$}, showing ${\sf QMA}={\sf QMA}(1-2^{-2^r},2^{-r})$ for any polynomial $r$ (\Cref{thm:QMA-small-completeness-error}). For any choice of $q$, our procedure uses just $O(1)$ calls to $V$, which seems somewhat remarkable, but is actually not new.
A previous observation \cite{nagaj2009FastAmpQMA,kobayashi2013stronger}, also used to show that ${\sf QCMA}={\sf QCMA}_1$ \cite{jordan2012QCMAisQCMA1}, is that if it is possible for the prover to communicate the \emph{exact} maximal acceptance probability, it is possible to achieve perfect completeness.
We discuss the result of truncating this procedure (i.e.~approximating the maximal acceptance probability), which we call \emph{probability truncation}, to get a ${\sf QMA}$ amplifier that similarly only makes $O(1)$ calls to the original verifier to improve the completeness error (but not the soundness) to $2^{-q}$ for any $q$. In this procedure however, required resources scale polynomially in $q$.

In the remainder of this section, we formally describe and analyze our ${\sf QMA}$ amplifier, and compare it with the state of the art \cite{fefferman2016errorReduction}, as well as probability truncation, which we also formalize (it was implicit in~\cite{kobayashi2013stronger}). A comparison between the different $\sf QMA$ amplifiers can be found in \Cref{tab:comparison}.

\paragraph{New ${\sf QMA}$ amplifier.} We first formally define what it means to truncate a ${\sf QMA}^\infty$ verifier.

\begin{definition}\label{def:truncate}
    Let $V$ be a ${\sf QMA}^\infty(c,s)$ verifier. Then we define its $D$-level truncation, $\tilde{V}_D$ (or $\tilde{V}$, if $D$ is clear from context) by modifying $V$ as follows. We replace the register $B$ with the register $\tilde B$ on $\mathrm{span}\{\ket{d}\}_{d=0}^{D-1}$. The increment $\incr$ is replaced with increment modulo $D$; and operations controlled on $\ket{0}_B$ are now controlled on $\ket{0}_{\tilde{B}}$.
\end{definition}

It is relatively easy to adapt the proof of either \Cref{prop:perfect completeness} or \Cref{thm:perfect completeness 2} to a truncated version, proving \Cref{thm:truncate} below, but it will be helpful to use the following lemma, which we prove in \Cref{sec:proof truncating witness} in order to show ${\sf QMA}={\sf QMA}^\infty$ (\Cref{thm:QMA-inf}).
\begin{lemma}\label{lem:truncate}
    Let $V$ be a ${\sf QMA}^\infty(c,s)$ verifier for $L$ with witness register $BW$ and ancilla register $A$, that makes $\ell$ applications of the controlled increment $\incr$.
    Let $m>\log \ell+3$ be bounded from above by a polynomial in $|x|$, and let $D$ be the smallest power of two greater than $2^{m}+\ell$. Let $\tilde{V}_D$ be the $D$-level truncation of $V$, and define its ancillary register $A'$ to consist of $A$, and the $\log(D)-m$ most significant qubits of $\tilde{B}$, and its witness register $W'$ to consist of $W$ and the remaining $m$ (least significant) bits of $\tilde{B}$. Then $\tilde{V}_D$ is a ${\sf QMA}(c',s)$ verifier for $L$ for some $c'=c-O(\sqrt{2^{-m}\ell})$.    \end{lemma}

\begin{proposition}
    \label{thm:truncate}
    Let $V(x)$ be a ${\sf QMA}(c,s)$ verifier for $L$, for some $s\in [0,1/4)$ and $c\in (1/4,1]$, with $l_M=\log(\dim(W))$ witness qubits; $l_A=\log(\dim(W))+\log(\dim(A))$ total qubits, and $t_A$ time complexity. Then for any $q\in \mathbb{N}$, there exists a ${\sf QMA}$ verifier $\tilde{V}(x)$ for $L$ with soundness $s'=1-(1-s)(1-4s)^2$ and completeness $c'=1-2^{-q}$, that uses the following resources:
    \begin{enumerate}
        \item $l_M'=l_M+\log\frac{q}{c-s}+O(1)$ qubits for the witness space;
        \item $l_A'=l_A+\log\frac{q}{c-s}+O(1)$ total qubits;
        \item 2 calls to $V$ and 1 call to $V^\dagger$; and
        \item $t_A'=O(t_A+\log\frac{q}{c-s})$ time complexity.
    \end{enumerate}
    The verifier $\tilde{V}$ uses the gates used by $V$, their inverses, and some gates that can exactly implement controlled $R_0$, and controlled increment modulo any power of 2.
\end{proposition}
We prove this theorem by truncating the second construction, from \Cref{sec:second-construction}. The properties of our verifier are summarized and compared with existing constructions in \Cref{tab:comparison}.
We can get a similar construction by truncating the first construction, from \Cref{sec:first-construction}. This has the advantage of only using \emph{1} call to each of $V$ and $V^\dagger$, but it requires $c-s$ to be constant.
\begin{proof}
    We can assume without loss of generality that $c$ and $s$ are centered around $\frac{1}{4}$ so that:
    $$c-\frac{1}{4} = \frac{1}{4}-s =: \delta,$$
    because if not, we can always construct a verifier with this property that accepts or rejects with some probability, and otherwise, determines whether to accept or reject by making a single call to $V$. Note that this implies $c\leq \frac{1}{2}$. We also have $\delta = \frac{c-s}{2}$.

    Let $\tilde{V}=\tilde{V'}_D$ be the $D$-level truncation of the ${\sf QMA}^\infty$ verifier $V'$ defined in \Cref{sec:second-construction}, for some $D$ to be specified shortly.
    Observe that this verifier uses witness space $l_M'=\log\dim\tilde{B}+\log\dim W = \log D+l_M$, and total space $l_A'=\log\dim\tilde{B}+\log\dim W+\log\dim A = \log D+l_A$. In terms of time, it consists of the following operations:
    \begin{itemize}
        \item two calls to $V$, and one to $V^\dagger$
        \item one controlled mod-$D$ increment, which costs $O(\log D)$ basic gates
        \item one controlled $R_0$, which costs $O(\log D)$ basic gates to control on being in $\ket{0}_{\tilde{B}}$, and $\log\dim A$ gates to implement $R_0$.
    \end{itemize}
    We can assume that $t_A=\Omega(\log\dim A)$, because otherwise some ancillary qubits are never touched, so we have $t_{A}'=O(t_A+\log D)$. We will eventually set $D=\Theta(\frac{q}{c-s})$, which yields the claimed complexities.

    \paragraph{Soundness.} We note that $V'$ makes $\ell=1$ calls to the controlled increment, so we can apply \Cref{lem:truncate} with $D= 2^m$ for some polynomial $m>4$.
    By~\Cref{thm:perfect completeness 2}, $V'$ has soundness $1-(1-s)(1-4s)^2$, and by \Cref{lem:truncate}, $\tilde{V}$ has the same soundness.

    \paragraph{Completeness.} If we appeal to \Cref{lem:truncate}, since $V'$ has completeness 1, we get completeness $c'=1-O(\sqrt{2^{-m}})=1-O(D^{-1/2})$. This is only better than previous $\sf QMA$-amplification procedures such as~\cite{fefferman2016errorReduction} when $c-s$ is very small. However, we can do better by truncating the witness from \Cref{thm:perfect completeness 2} directly. Consider the acceptance probability $\tilde{p}_{\text{acc}}(\tilde\psi)$ of $\tilde{V}$ on a state of the form
    $$\ket{\tilde\psi}_{\tilde B}\ket{w}_W,\qquad \ket{\tilde\psi}=\sqrt{\frac{1-\gamma^2}{1-\gamma^{2(D-1)}}}\sum_{d=0}^{D-2}\gamma^d\ket{d}, \qquad \gamma=1-\frac{1}{2p},$$
    for $\ket{w}$ an optimal witness for $V$, with acceptance probability $p\geq c$. This is just a truncated version of $\ket{\psi}_B\ket{w}$, for $\ket{\psi}_B$ as in \cref{eq:inf-witness}.
    Since $\ket{\tilde\psi}$ is not supported on $\ket{D-1}_{\tilde B}$, and we only make one increment to this register, the increment modulo $D$ always behaves identically to the increment $\incr$, and so in particular, the acceptance probability $\tilde p_{\text{acc}}(\tilde\psi)$ of $\tilde{V}$ on input $\ket{\tilde\psi}\ket{w}$, is the same as the acceptance probability $p_{\text{acc}}(\tilde\psi)$ of $V'$ on input $\ket{\tilde\psi}\ket{w}$. We have:
    $$\tilde p_{\text{acc}}(\tilde{\psi}) = p_{\text{acc}}(\tilde{\psi}) \geq p_{\text{acc}}(\psi)-T(\tilde\psi,\psi),$$
    where $p_{\text{acc}}(\psi)=1$ is the acceptance probability of $V'$ on input $\ket{\psi}\ket{w}$, and $T(\tilde\psi,\psi)$ is the trace distance between $\ket{\tilde\psi}\ket{w}$ and $\ket{\psi}\ket{w}$. By \cref{eq:inf-witness},
    $$\braket{\tilde\psi}{\psi} = \frac{1-\gamma^2}{\sqrt{1-\gamma^{2(D-1)}}}\sum_{d=0}^{D-2}\gamma^{2d} = \sqrt{1-\gamma^{2(D-1)}}.$$
    Thus, we have
    $$T(\tilde\psi,\psi) = \sqrt{1-|\braket{\tilde\psi}{\psi}|^2}=\gamma^{D-1}.$$
    \noindent Since $c\in (1/4,1/2]$, we have:
    $$|\gamma| = \abs{\frac{1-2p}{2p}} \leq \abs{\frac{\frac{1}{2}-2\delta}{\frac{1}{2}+2\delta}}=\abs{1-\frac{4\delta}{\frac{1}{2}+2\delta}}\leq \abs{1-4\delta}.$$
    Thus
    $$\tilde{p}_{\text{acc}}(\tilde\psi)\geq 1-\left(1-4\delta\right)^{D-1}\geq 1-e^{-4\delta(D-1)},$$
    so there is some choice of $D=\Theta(q/\delta)=\Theta(\frac{q}{c-s})$ such that $\tilde{p}_{\text{acc}}\geq 1-2^{-q},$ establishing the claimed completeness of $c'=1-2^{-q}$.
\end{proof}

\noindent \Cref{thm:QMA-small-completeness-error}, restated here in slightly more detail, follows as a corollary of \Cref{thm:truncate}.
\vskip10pt
\noindent\textbf{Theorem~\ref*{thm:QMA-small-completeness-error}.} \emph{For any polynomial $r$, ${\sf QMA}={\sf QMA}(1-2^{-2^r},2^{-r})$,
if ${\sf QMA}(1-2^{-2^r},2^{-r})$ is defined with respect to a gate set that includes exact inverses, and allows for exact implementation of controlled increment and controlled $R_0$.}
\vskip10pt

Note that we cannot assume that ${\sf QMA}(1-2^{-2^r},2^{-r})$ is gate-set independent, since mapping from one gate set to another may introduce $\omega(2^{-2^r})$ error.

\begin{proof}[Proof of \Cref{thm:QMA-small-completeness-error}]
    Given a ${\sf QMA}={\sf QMA}(1-2^{-r'},2^{-r'})$ verifier for some polynomial $r'$, we can apply the construction in \Cref{thm:truncate} with $q=2^{r}$ to get a polynomial-time verifier $\tilde{V}$ with completeness $c'=1-2^{-2^{r}}$, and soundness
    $$s'\leq 9s \leq 2^{-r'+\log(9)}.$$
    Taking $r'>r+\log(9)$ gives $s'\leq 2^{-r}$, as desired.
\end{proof}

\paragraph{Probability truncation $\sf QMA$ amplifier.} Let us also compare with the \emph{probability truncation} amplifier, based on~\cite{nagaj2009FastAmpQMA,kobayashi2013stronger}, that,
like our amplifier, amplifies the completeness, but not the soundness. This procedure is similar to the one used to show ${\sf QCMA}={\sf QCMA}_1$, where ${\sf QCMA}$ is like ${\sf QMA}$, except that the witness is restricted to be classical \cite{jordan2012QCMAisQCMA1}. In that case, it may be assumed that the acceptance probability has a polynomial-length description, and can be given as part of the witness. Given a verifier, a classical witness, and some $p>1/2$ that is claimed to be its acceptance probability, it is possible to check this claim using 2 calls to the verifier and its inverse. In ${\sf QMA}$, we cannot assume that the acceptance probability has an efficient (or finite) description, but even giving a $q$-bit approximation of the success probability can amplify completeness. We sketch such a protocol here, based on the variant in \cite{kobayashi2013stronger}, where this idea was used to show perfect completeness in case the prover and verifier pre-share some EPR pairs, as well as for interactive proofs with an odd number of rounds.
This approach is similar in spirit to our work, in the sense that the completeness is enhanced by extending the witness to contain information about the acceptance probability of the original verifier. However, as we will see, probability truncation only reaches a completeness exponentially close to 1 while keeping the verifier polynomial.

We now describe the modified verifier, given a ${\sf QMA}$ verifier with completeness $c > \frac12$ and soundness $s < \frac12$.
The witness of the modified verifier is on two registers, $P$ and the original witness register $W$.
The verifier measures the $q$-qubit register $P$, and interprets the resulting bitstring as a probability $p \in [0,1]$. If $p < \frac12$, the verifier rejects.
Otherwise, the verifier runs the following circuit
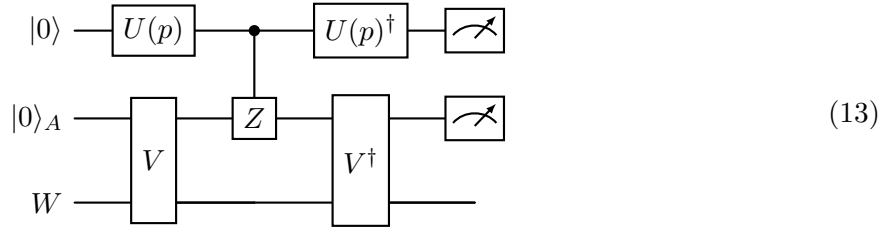
\begin{equation}\label{eq:probability truncation}
    \begin{quantikz}
        \lstick{$\ket{0}$} & \gate{U(p)} & \ctrl{1} & \gate{U(p)^\dagger} & \meter{} \\
        \lstick{$\ket{0}_A$} & \gate[2]{V} & \gate{Z} & \gate[2]{V^{\dagger}} & \meter{} \\
        \lstick{$W$} &  & \qw & & \qw
    \end{quantikz}
\end{equation}
where $U(p)$ is the single-qubit unitary
\begin{align*}
    U(p) = \begin{pmatrix}
               \sqrt{1 - \gamma} & \sqrt{\gamma} \\ \sqrt{\gamma} & -\sqrt{1-\gamma}
           \end{pmatrix} \qquad \text{ for } \quad \gamma = \frac{1}{2p}.
\end{align*}
The verifier rejects if they measure both registers in the initial states $\ket{0}$ and $\ket{0}_A$.
It can be shown that if the witness has the form $\ket{p}_P \ket{w}_W$ with $\ket{w}$ the optimal witness for $V(x)$ with acceptance probability $p > \frac12$ and $\ket{p}_P$ \emph{exactly} encodes this acceptance probability, then $V'(x)$ accepts with certainty.
For soundness, if the maximal acceptance probability of $V(x)$ is at most $s = \frac12 - \delta < \frac12$, then Proposition 18 in \cite{kobayashi2013stronger} shows that for any $\gamma$, Arthur rejects with probability at least $p_{\text{reject}} \geq 4\delta^2$, which gives a new soundness of $s' = 4s(1 - s)$.
In general, the maximal acceptance probability may not have an exact expression with a polynomial number of bits. In that case, Merlin can send a $q$-bit approximation of $p$, which gives an approximation of $O(2^{-q})$ to the value of $\gamma$ that would give perfect completeness.
Thus, we obtain completeness $1-O(2^{-q})$. The new verifier uses $q$ extra (qu)bits for the witness and verifier (the full verifier controls on $\ket{p}_P$) and at least $O(q+l_M)=O(q+t_A)$ extra steps to read the bits of $\ket{p}_P$, and also to check if $A$ is in the state~$\ket{0}$.

Finally, we briefly relate our results to the quantum oracle separation between ${\sf QMA}$ and ${\sf QMA}_1$ shown in \cite{aaronson09perfectCompQMA}. Aaronson notes that the argument no longer holds when using infinite space. However, the argument \emph{does} apply to the truncations described in this section, and shows that there is no straightforward way to modify the ${\sf QMA}_1$ verifier to a finite dimensional variant that preserves perfect completeness. An interesting open question is whether the argument in \cite{aaronson09perfectCompQMA}, based on analyticity of the maximal acceptance probability with respect to the quantum oracle, can establish a quantitative bound on the gap to perfect completeness (for example, proving obstacles to improving further on \Cref{thm:QMA-small-completeness-error}).

\begin{table}[tbp]
    \centering
    \noindent\makebox[\textwidth][c]{\begin{minipage}{1.15\textwidth}
            \centering
            \begingroup
            \setlength{\tabcolsep}{3pt}
            \begin{tabular}{c|c|c|m{1.25cm}|c|c|c}
                                                   & $1-c'$    & $s'$              & calls to $V$ \& $V^\dagger$ & total time $t_A'$                            & $l_M'$                        & $l_A'$                                      \\
                \hline
                \cite{fefferman2016errorReduction} & $2^{-q}$  & $2^{-q}$          & $O(\frac{q}{c-s})$          & $O(\frac{q}{c-s})t_A$                        & $l_M$                         & $l_A+O(\log\frac{q}{c-s})$                  \\
                \cite{kobayashi2013stronger}       & $ 2^{-q}$ & $4s(1 - s)$       & \centering $2$              & $O(t_A + q)$                                 & $l_M + O(q)$                  & $l_A + O(q)$                                \\
                New                                & $2^{-q}$  & $1-(1-s)(1-4s)^2$ & \centering 3                & $O(t_A+\log \frac{q}{c-s})$                  & $l_M+\log \frac{q}{c-s}+O(1)$ & $l_A+O(\log \frac{q}{c-s})$                 \\
                New+                               & $2^{-q}$  & $2^{-q'}$         & $O(\frac{q'}{c-s})$         & $O\!\!\left(\frac{q'}{c-s}t_A+\log q\right)$ & $l_M+\log q+O(1)$             & $O\!\!\left(l_A+\log\frac{qq'}{c-s}\right)$ \\
            \end{tabular}
            \endgroup
        \end{minipage}}
    \caption{Comparison of our amplification procedure with prior work. The row ``New'' refers to the construction in \Cref{thm:truncate}, and
    ``New+'' refers to the construction obtained by first applying \cite{fefferman2016errorReduction}, and then applying \Cref{thm:truncate}.
    The amplified completeness and soundness are $c'$ and $s'$ respectively. The memory used by Merlin (i.e. the witness)  in the original and amplified protocols are $l_M$ and $l_M'$, respectively. The memory and time used by the verifier (Arthur) in the original protocol is $l_A$ and $t_A$, and in the amplified protocols is $l_A'$ and $t_A'$. Note that both \cite{kobayashi2013stronger} and ``New'' result in a verifier with a \emph{worse} soundness-completeness gap than before: $c'-s'=\Theta((c'-s')^2-2^{-q})$, where \cite{fefferman2016errorReduction} (and ``New+'') increases the gap to almost 1. We have made the completeness the same in all constructions, for comparison purposes, but note that in our two constructions, the dependence on $q$ is logarithmic, so we may take $q=2^{r}$ for any polynomial $r$. This is in contrast to previous work, where $q$ must be polynomial in order for $V'$ to be efficient.}
    \label{tab:comparison}
\end{table}

\section{Proof of \texorpdfstring{\Cref{thm:QMA-inf}}{Theorem 3.2}}\label{sec:proof truncating witness}

To prove \Cref{thm:QMA-inf}, which states that ${\sf QMA}={\sf QMA}^\infty$,
we need to show that a verifier that uses an infinite counter can be truncated to only use a finite-dimensional system. The challenge is that the witness state is unrestricted on the counter, and it is a priori unclear that one can truncate the counter. The idea will be that if the verifier uses in total at most $\ell$ increment operations, and there is a set of consecutive counter states $\ket{d}, \ket{d+1}, \dots, \ket{d+\ell}$ on which the witness state has no support, then the witness can be broken down into separate parts that do not interact.
The idea is now to approximate an arbitrary witness state by a state that has such intervals, for which the following lemma is helpful.

\begin{lemma}\label{lem:remove small consecutive subspace}
    Fix integers $\ell < D$. Consider a Hilbert space $\CC^D$, and for $k = 0, \dots, D - \ell$ let
    \begin{align*}
        \HH_{[0,k-1],[k+\ell,D-1]} = \span \, \{ \ket{d}, \, d \in \{0, \dots, D - 1 \} \setminus \{k, k + 1, \dots, k + \ell - 1\} \}.
    \end{align*}
    Then, for any state $\ket{\psi} \in \CC^D$, there exists $k$ and a state $\ket{\phi_k} \in \HH_{[0,k-1],[k+\ell,D-1]}$ such that
    \begin{align*}
        F(\phi_k, \psi) \geq \sqrt{1 - \frac{\ell}{D - \ell}},
    \end{align*}
\end{lemma}

\begin{proof}
    Let $\Pi_k$ denote the projection onto $\HH_{[0,k-1],[k+\ell,D-1]}$.
    If we expand $\ket{\psi} = \sum_{d=0}^{D-1} {\psi_d} \ket{d}$, we have
    \begin{align*}
        (D - \ell + 1) \min_k \norm{(I - \Pi_k) \ket{\psi}}^2 & \leq \sum_{k = 0}^{D - \ell} \norm{(I - \Pi_k) \ket{\psi}}^2  = \sum_{k = 0}^{D - \ell} \sum_{d = k}^{k + \ell - 1} \abs{\psi_d}^2 \leq \ell \sum_{d = 0}^D \abs{\psi_d}^2 = \ell
    \end{align*}
    so there exists $k$ such that $\norm{(I - \Pi_k) \ket{\psi}}^2 \leq \ell / (D - \ell)$. We then take $\ket{\phi_k} = \Pi_k \ket{\psi} / \norm{\Pi_k \ket{\psi}}$.
\end{proof}

The main work of this section is to prove \Cref{lem:truncate}, which we restate here for convenience.

\vskip10pt
\noindent\textbf{Lemma~\ref*{lem:truncate}.} \emph{Let $V$ be a ${\sf QMA}^\infty(c,s)$ verifier for $L$ with witness register $BW$ and ancilla register $A$, that makes $\ell$ applications of the controlled increment $\incr$.
    Let $m>\log \ell+3$ be bounded from above by a polynomial in $|x|$, and let $D$ be the smallest power of two greater than $2^{m}+\ell$. Let $\tilde{V}_D$ be the $D$-level truncation of $V$, and define its ancillary register $A'$ to consist of $A$, and the $\log(D)-m$ most significant qubits of $\tilde{B}$, and its witness register $W'$ to consist of $W$ and the remaining $m$ (least significant) bits of $\tilde{B}$. Then $\tilde{V}_D$ is a ${\sf QMA}(c',s)$ verifier for $L$ for some $c'=c-O(\sqrt{2^{-m}\ell})$.}
\vskip10pt

An immediate corollary is the following.
\begin{corollary}\label{cor:QMA1-inf-is-QMA-inf-is-QMA}
    For all polynomial $r$, ${\sf QMA}^\infty(1,2^{-r})={\sf QMA}^\infty\left(\frac{1}{2}+\frac{1}{r},\frac{1}{2}-\frac{1}{r}\right)={\sf QMA}$.
\end{corollary}
\begin{proof}
    By~\Cref{cor:QMAinQMA-inf-exp-soundness}, ${\sf QMA}\subseteq {\sf QMA}^\infty(1,2^{-r})$. Together with the obvious inclusion, we get:
    $${\sf QMA}\subseteq {\sf QMA}^\infty(1,2^{-r})\subseteq {\sf QMA}^\infty\left(\frac{1}{2}+\frac{1}{r},\frac{1}{2}-\frac{1}{r}\right).$$
    Thus, we complete the proof by showing ${\sf QMA}^\infty(1/2+1/r,1/2-1/r)\subseteq {\sf QMA}$.
    Let $V$ be a ${\sf QMA}^\infty(c,s)$ verifier for $L\in {\sf QMA}^\infty(c,s)$, with $c=1/2+1/r$ and $s=1/2-1/r$. Then by \Cref{lem:truncate}, using $m=\log\ell+2r'$ for some polynomial $r'$,
    $$L\in {\sf QMA}(c-O(2^{-r'}),s)\subseteq {\sf QMA}\left(\frac{1}{2}+\frac{1}{r}-2^{-(r'+O(1))},\frac{1}{2}-\frac{1}{r}\right)={\sf QMA},$$
    so ${\sf QMA}^\infty(1/2+1/r,1/2-1/r)\subseteq{\sf QMA}$.
\end{proof}

\noindent\Cref{cor:QMA1-inf-is-QMA-inf-is-QMA} establishes:
\begin{itemize}
    \item \Cref{thm:QMA-inf}: Taking $r=6$, ${\sf QMA}={\sf QMA}^\infty(\frac{1}{2}+\frac{1}{6},\frac{1}{2}-\frac{1}{6})={\sf QMA}^\infty(2/3,1/3)={\sf QMA}^\infty$;
    \item \Cref{thm:QMAisQMA1-inf}: Taking $r=\log(3)$, ${\sf QMA}={\sf QMA}^\infty(1,2^{-\log(3)})={\sf QMA}^\infty(1,1/3)={\sf QMA}_1^\infty$;
    \item \Cref{thm:QMA-inf-with-different-c-s}:
          To show ${\sf QMA}_1^\infty={\sf QMA}^{\infty}\left(1,1-\frac{1}{r}\right)$,
          we will make use of the fact that ${\sf QMA}^\infty(1,1-1/r)\subseteq {\sf QMA}^{\infty}(\frac{1}{2}+\frac{1}{4r-2},\frac{1}{2}-\frac{1}{4r-2})={\sf QMA}$ (by \Cref{cor:QMA1-inf-is-QMA-inf-is-QMA}). The inclusion is established as follows. Let $V$ be any ${\sf QMA}^{\infty}(1,1-1/r)$ verifier. Define $V'$ so that with probability $\frac{1}{2}-\frac{1}{4r-2}$ it rejects, and otherwise it runs $V$. It is easy to verify $V'$ is a ${\sf QMA}^{\infty}(\frac{1}{2}+\frac{1}{4r-2},\frac{1}{2}-\frac{1}{4r-2})$ verifier.
          Then combining this with the obvious inclusions gives:
          $${\sf QMA}_1^\infty={\sf QMA}^{\infty}\left(1,\frac{1}{3}\right)\subseteq{\sf QMA}^\infty\left(1,1-\frac{1}{r}\right)\subseteq 
              {\sf QMA}={\sf QMA}_1^\infty.$$ 
          The remainder of \Cref{thm:QMA-inf-with-different-c-s} is straightforward:
          \begin{itemize}
              \item[$\bullet$] ${\sf QMA}_1^\infty={\sf QMA}={\sf QMA}^{\infty}(1,2^{-r})$
              \item[$\bullet$] ${\sf QMA}^\infty\subseteq {\sf QMA}^\infty(1-2^{-r},2^{-r}) \subseteq{\sf QMA}^\infty(1,2^{-r})={\sf QMA}={\sf QMA}^\infty$, and therefore ${\sf QMA}^\infty={\sf QMA}^\infty(1-2^{-r},2^{-r})$
              \item[$\bullet$] ${\sf QMA}^\infty={\sf QMA}={\sf QMA}^\infty\left(\frac{1}{2}+\frac{1}{r},\frac{1}{2}-\frac{1}{r}\right)$.
          \end{itemize}
\end{itemize}

\noindent Having established these consequences, it remains only to prove \Cref{lem:truncate}.
\begin{proof}[Proof of \Cref{lem:truncate}]
    We introduce the following notation: let $\HH_d$ be the subspace spanned by states of the form $\ket{d}_B \ket{\psi}_{AW}$ and for $d_1<d_2$ define
    \begin{align*}
        \HH_{[d_1,d_2]} = \bigoplus_{d = d_1}^{d_2} \HH_d.
    \end{align*}

    Recall from \Cref{def:truncate} that $\tilde{V}_D$ is obtained from $V$ by replacing $B$ with $\tilde{B}$ on $\HH_{[0,D-1]}$, and the controlled increment $\incr$ with increment modulo $D$.
    By making the $\log(D)-m$ most significant bits of $\tilde{B}$ part of the ancilla,
    we ensure that the $\tilde B$-part of the witness is restricted to $\HH_{[0,2^m-1]}$.
    Since $V$ only uses $\ell$ applications of the increment, $\tilde B$ is only ever supported on $d<2^m+\ell\leq D-1$,
    so no matter what witness is used, there is no distinction between the increment $\incr$,
    and increment modulo $D$, and so $\tilde{V}_D$ is exactly equivalent to the untruncated verifier $V$ when it uses a witness state with $B$ restricted to $\HH_{[0,2^m-1]}$. This immediately implies soundness $s$ of $\tilde{V}$, since we are simply restricting the allowed witness states as compared with $V$.

    For completeness, it suffices to show that if there exists a witness state $\ket{w}$ with acceptance probability $p\geq c$, there exists a witness state $\ket{w'} \in \HH_{[0,2^m-1]}$ that accepts with probability close to $p$.
    To this end, we consider an arbitrary witness state
    \begin{align*}
        \ket{w} = \sum_{d=0}^{\infty} w_d \ket{\phi_d} , \qquad \ket{\phi_d} \in \HH_d.
    \end{align*}
    We now show that there exists an approximation $\ket{\tilde w}$ of $\ket{w}$ such that $\mathbb{N}$ consists of infinitely many intervals such that: the support of $\ket{\tilde{w}}$ is contained in the even intervals (numbering them from 0), the even intervals have length $\leq 2^{m-2}$, and the odd intervals have length $\ell$.
    To this end, we write
    \begin{align*}
        \ket{w} = \sum_{j = 0}^\infty \sqrt{\beta_j} \ket{\theta_j}, \qquad \ket{\theta_j} \in \HH_{[j 2^{m-3}, (j+1) 2^{m - 3} - 1]}.
    \end{align*}
    We may now apply \Cref{lem:remove small consecutive subspace} to obtain states $\ket{\tilde \theta_j}$ that are not supported on some subspace $\HH_{[k_j, k_j+\ell-1]}$ for some $k_j\in [j 2^{m-3},(j+1) 2^{m - 3} - \ell]$, and which have fidelity with $\ket{ \theta_j}$ bounded as $F(\theta_j, \tilde \theta_j) \geq \sqrt{1 - \ell(2^{m - 3} - \ell)^{-1}}$. We now define $\ket{\tilde w} = \sum_j \sqrt{\beta_j} \ket{\tilde \theta_j}$, which has distance to $\ket{w}$ bounded by
    \begin{align}\label{eq:T}
        T(w, \tilde w) \leq \sqrt{1 - F(w, \tilde w)^2} = \sqrt{1 - \sum_j \beta_j F(\theta_j, \tilde \theta_j)^2} \leq \sqrt{\frac{\ell}{2^{m-3} - \ell}}.
    \end{align}
    The figure below shows the first part $\mathbb{N}$, with the support of each $\ket{\theta_j}$ an equal sized block of $2^{m-3}$. Each such block has some length-$\ell$ interval -- the $[k_j,k_j+\ell-1]$ -- removed, shown in grey. The support of $\ket{\tilde w}$ is thus within the white parts of the number line, which naturally decomposes into white intervals, the largest of which has length at most $2(2^{m-3}-\ell)=2^{m-2}-2\ell$.
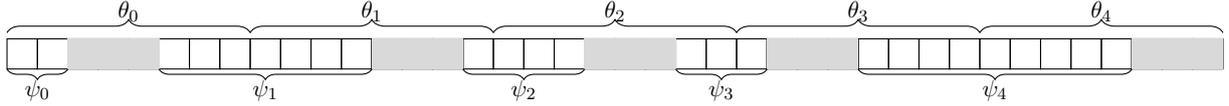
    \begin{center}
        \begin{tikzpicture}[scale=0.4, every node/.style={scale=0.8}]
            \foreach \x in {0,...,39} {
                    \draw (\x, 0) rectangle ++(1, 1);
                }

            \foreach \x in {2,3,4, 12,13,14, 19,20,21, 25,26,27, 37,38,39} {
                    \fill[gray!30] (\x, 0) rectangle ++(1, 1);
                }

            \foreach \i in {0,1,2,3,4} {
                    \draw[decorate,decoration={brace,amplitude=5pt}]
                    ({\i*8}, 1.2) -- ({\i*8 + 8}, 1.2)
                    node[midway, yshift=10pt] {$\theta_{\i}$};
                }

            \foreach \start/\fin/\i in {0/2/0, 5/12/1, 15/19/2, 22/25/3, 28/37/4} {
                    \draw[decorate,decoration={brace,amplitude=4pt,mirror}]
                    (\start, 0) -- (\fin, 0)
                    node[midway, yshift=-10pt] {$\psi_{\i}$};
                }
        \end{tikzpicture}
    \end{center}
    In other words, the state $\ket{\tilde w}$ can be written as
    \begin{align*}
        \ket{\tilde w} = \sum_{i = 0}^{\infty} \sqrt{\alpha_i} \ket{\psi_i}, \qquad \ket{\psi_i} \in \HH_{[d_i, d_i + D_i - 1]}
    \end{align*}
    for some natural numbers $\{d_i,D_i\}_{i=0}^\infty$ such that
    $d_0 = 0$; $\forall i\in \mathbb{N}$, $d_{i+1} \geq d_i + D_i + l - 1$; and $\forall i\in \mathbb{N}$, $D_i \leq 2^{m-2}-2\ell$. The diagram above shows the length-$D_i$ support of each $\ket{\psi_i}$.

    The acceptance probability using the witness $\ket{\tilde w}$ is given by
    \begin{align*}
        p_{\accept}(\tilde w) = \norm{\Pi_{\accept} \tilde{V}(x) \ket{0}_A \ket{\tilde w}}^2=\norm{\Pi_{\accept} V(x) \ket{0}_A \ket{\tilde w}}^2 = \Big\lVert \sum_i \sqrt{\alpha_i} \underbrace{\Pi_{\accept} V(x) \ket{0}_A \ket{\psi_i}}_{=:\ket{\xi_i}} \Big\rVert^2.
    \end{align*}
    We now note that $\Pi_{\accept}$ preserves the spaces $\HH_d$, and since the verifier uses at most $\ell$ (controlled) increments, we have that $V(x) \ket{0}_A \ket{\psi_i} \in \HH_{[d_i, d_i + D_i + \ell - 1]}$.
    This implies that the (unnormalized) states $\ket{\xi_i}$ are orthogonal, so
    \begin{align*}
        p_{\accept}(\tilde w) = \sum_i \alpha_i \norm{\ket{\xi_i}}^2.
    \end{align*}
    Next, we note that for $i > 0$, starting from $\ket{\psi_i}$, the state is fully supported on the subspace $\ket{d}$ for $d \geq 1$, so the controlled $R_0$ is applied deterministically.
    From this, it is easy to see that $\norm{\ket{\xi_i}}^2$ is the same if we uniformly shift the counter register without reaching the $d = 0$ state, so
    \begin{align*}
        \norm{\ket{\xi_i}}^2 = \norm{\Pi_{\accept} V(x) \ket{0} (\incr^\dagger)^d \ket{\psi_i}}^2
    \end{align*}
    for any $d < d_i$.
    We can bound
    \begin{align}\label{eq:bound-truncated-acceptance-probability}
        p_{\accept}(\tilde w) \leq \alpha_0 \norm{\ket{\xi_0}}^2 + (1 - \alpha_0) \max_{i \geq 1} \norm{\ket{\xi_{i}}}^2
    \end{align}
    which corresponds to the acceptance probability of a state $\ket{w'}$
    \begin{align*}
        \ket{w'} = \sqrt{\alpha_0} \ket{\psi_0} + \sqrt{1 - \alpha_0} \ket{\psi'} \qquad
    \end{align*}
    where $\ket{\psi'}=(\incr^\dagger)^d\ket{\psi_i}$ for $i$ the maximizer in \cref{eq:bound-truncated-acceptance-probability}, and $d=d_i-D_0-\ell$
    so that $\ket{\psi'}\in \HH_{[D_0 + \ell, D_0 + \ell + D_i - 1]}$.
    In particular, $\ket{w'} \in \HH_{[0, D_0 + \ell + D_i]} \subseteq \HH_{[0, 2^{m -1} -2\ell+ \ell]}\subseteq \HH_{[0,2^m-1]}$.
    We conclude that, using \cref{eq:T}:
    \begin{align*}
        p_{\accept}(w') \geq p_{\accept}(\tilde w) \geq p_{\accept}(w) - T(w, \tilde w) \geq c - \bigO{\sqrt{2^{-m}\ell}}.
    \end{align*}

    To complete the proof, we need to argue that $\tilde{V}_D$ runs in polynomial time. It consists of $O(1)$ calls to $V$ and $V^\dagger$, each of which takes time bounded by some polynomial $t_A$; an increment modulo $D$, which takes time $O(\log D)$, and a reflection $R_0$ on $\log\dim A \leq l_A$ qubits, controlled on $\log(D)$ bits, which takes time $O(\log(D)+l_A)$. Thus, we can upper bound the running time of $\tilde{V}_D$ by
    $$\tilde{t}_A=O(t_A+\log(D)+l_A).$$
    We can assume $l_A=O(t_A)$, as otherwise there are ancilla qubits that the verifier never uses. Moreover, $D=\Theta(2^m+\ell)$, so since $m$ is polynomial, so is $\log(D)$. Thus $\tilde{t}_A$ is polynomial in $|x|$.
\end{proof}

\begin{remark}\label{rmk:extensions}
    For the sake of simplicity, and to minimally enhance the computational power of the verifier, we defined ${\sf QMA}$ with a \emph{single} counter that can only \emph{increment}.
    However, it is not hard to strengthen the proof of \Cref{lem:truncate} to allow for multiple tapes, or for a counter that can also decrement, for which we now sketch the arguments.
    Let us start with the second case. In order to make the counter invertible, we allow a two-sided infinite system with states $\ket{d}$ for $d \in \ZZ$. The control on $B$ is still on $d \geq 1$.
    We can now allow controlled increments and decrements. The truncation argument is modified as follows. If there are $l$ incrementing or decrementing operations, it now suffices to approximate the state with non-support intervals of length $2l$. We get states $\ket{\psi_i}$ in the same way, but now we separately treat states with support on the interval around zero, the states with $d > 0$ and the states with $d < 0$.
    For the case with multiple counters $B_1, \dots, B_m$ for polynomial $m$, (where we assume each $B_i$ interacts with a set of qubits as before), we note that the argument in \Cref{lem:truncate} does not use that $\HH_d$ is finite-dimensional. So, that means that we can use \Cref{lem:truncate} to truncate the $B_i$ one by one.
\end{remark}

\section{Bosonic quantum computers}\label{sec:bosonic quantum computers}
A good reason to consider infinite dimensional Hilbert spaces is that many physical systems are in fact modelled by infinite dimensional Hilbert spaces, and restricting to finite dimensional subspaces is only an approximation.
This is in particular the case for bosonic systems, which are most prominent in quantum optics.
Here, the Hilbert space is a tensor product of spaces $\HH_x$ for $x \in S$.
Each $\HH_x$ has an orthonormal basis $\ket{d}$ for $d \in \ZZ_{\geq 0}$, where $d$ is called the \emph{particle number}.
We define the energy by the Hamiltonian
\begin{align*}
    H_x = \sum_{d=0}^{\infty} d \ketbra{d}{d}_x \ot I_{S \setminus x}, \qquad H = \sum_{x \in S} H_x.
\end{align*}
Up to rescalings and constant factors we can ignore for our purpose, this corresponds to the standard Hamiltonians for electromagnetic modes in quantum optics.
In this setting, gate sets are typically defined by time evolution along Hamiltonians that are polynomials in the creation and annihilation operators
\begin{align*}
    a \ket{d} = \sqrt{d}\ket{d - 1} \quad \text{ and } \quad a^\dagger \ket{d} = \sqrt{d+1}\ket{d+1}.
\end{align*}
The increment operator $\incr$ is not easily constructed this way.

The status of variants of ${\sf BQP}$ or ${\sf QMA}$ in this bosonic setting has been studied, but their precise relation to the standard notions in unclear \cite{chabaud2025bosonic,upreti2025bounding}.
If one takes a gate set of operations $\{U_i\}$ that increase the energy at most polynomially, then this gives no greater computation power than ${\sf BQP}$. If we define ${\sf QMA}$ with these operations, and additionally constrain the witness to have at most polynomial energy, this also does not increase the computational power \cite{arzani2025can}.

Our results suggest several interesting questions: Can one achieve perfect completeness for ${\sf QMA}$ with a physically reasonable bosonic gate set? What happens in this setting to ${\sf QMA}$ when we do not constrain the witness? It seems possible that a similar argument as in \Cref{thm:QMA-inf} goes through, but we leave this for future work.

\subsection*{Acknowledgments}

We thank Scott Aaronson, Simon Apers, Aleksandrs Belovs, Vedran Dunjko, David Gosset and Tamara Kohler for useful comments and discussions about the ideas in this paper.

This work is co-funded by the European Union (ERC, ASC-Q, 101040624); and Divide \& Quantum  (with project number 1389.20.241) of the research programme NWA-ORC, which is (partly) financed by the Dutch Research Council (NWO). SJ is a CIFAR Fellow in the Quantum Information Science Program.

\bibliographystyle{alpha}
\bibliography{refs}

\end{document}